\documentclass[11pt]{article}

\usepackage[utf8]{inputenc} 
\usepackage[english]{babel} 
\usepackage{fullpage}

\usepackage{amsfonts} 
\usepackage{amsmath} 
\usepackage{amssymb}
\usepackage{amsthm} 
\usepackage{multirow}
\usepackage{longtable}

\usepackage{algorithm}
\usepackage[noend]{algpseudocode}
\usepackage{epsfig,color}
\usepackage{colortbl}
\usepackage{tikz}
\usetikzlibrary{mindmap}
\usepackage[normal]{subfigure}
\usepackage{tabu}

\newtheorem{rmk}{Remark}
\newtheorem{prop}{Property}
\newcommand{\nqed}{}

\makeatletter
\def\@cline#1-#2\@nil{%
  \omit
  \@multicnt#1%
  \advance\@multispan\m@ne
  \ifnum\@multicnt=\@ne\@firstofone{&\omit}\fi
  \@multicnt#2%
  \advance\@multicnt-#1%
  \advance\@multispan\@ne
  \leaders\hrule\@height\arrayrulewidth\hfill
  \cr
  \noalign{\nobreak\vskip-\arrayrulewidth}}
\makeatother
  
\makeatletter
\newcommand\nobreakhline{%
\multispan\LT@cols
\unskip\leaders\hrule\@height\arrayrulewidth\hfill\\*}
\newcommand\nobreakcline[1]{\@nobreakcline#1\@nil}%
\def\@nobreakcline#1-#2\@nil{%
  \omit
  \@multicnt#1%
  \advance\@multispan\m@ne
  \ifnum\@multicnt=\@ne\@firstofone{&\omit}\fi
  \@multicnt#2%
  \advance\@multicnt-#1%
  \advance\@multispan\@ne
  \leaders\hrule\@height\arrayrulewidth\hfill\\*
  \noalign{\vskip-\arrayrulewidth}}
\makeatother

\title{A Bit-Parallel Russian Dolls Search for a Maximum Cardinality Clique in a
Graph\thanks{This work has been partially supported by the Stic/AmSud
joint program by CAPES (Brazil), CNRS and MAE (France), CONICYT (Chile) and
MINCYT (Argentina) -- project STABLE -- and the Pronem program by FUNCAP/CNPq
(Brazil) -- project ParGO. The first author was with Universidade Federal do
Cear\'a, Departamento de Computação, Brazil, when this work was done. The second author has been partially granted by the
"Pesquisador Visitante Especial" (CNPq program) process 313831/2013-0.}}

\author{\sc{Ricardo C. Corrêa} \\ {\small Universidade Federal Rural do Rio de
Janeiro,} \\
{\small Departamento de Ciência da Computação,} \\
{\small Av. Governador Roberto Silveira S/N,}  \\
{\small 26020-740 Nova Iguaçu - RJ, Brazil} \and 
\sc{Philippe Michelon} \\ {\small Université d’Avignon et des Pays du Vaucluse,}\\
{\small Laboratoire d’Informatique d’Avignon,} \\ {\small F-84911 Avignon, Cedex
9, France} \and 
\sc{Bertrand Le Cun}, \sc{Thierry Mautor} \\ {\small Université de
Versailles Saint Quentin,} \\
{\small 45 Avenue des Etats Unis,} \\ {\small 78035 Versailles, France} \and
\sc{Diego Delle Donne} \\ {\small Sciences Institute,} \\ {\small National
University of General Sarmiento,} \\ {\small J. M. Gutiérrez 1150, Malvinas
Argentinas,} \\ {\small (1613) Buenos Aires, Argentina}}

\begin{document}

\maketitle

\begin{abstract}
Finding the clique of maximum cardinality in an arbitrary graph is an NP-Hard
problem that has many applications, which has motivated studies to
solve it exactly despite its difficulty. The great majority of algorithms
proposed in the literature are based on the Branch and Bound method. In this
paper, we propose an exact algorithm for the maximum clique problem based on the Russian Dolls Search method. 
When compared to Branch and Bound, the main difference of the Russian Dolls
method is that the nodes of its search tree correspond to decision subproblems,
instead of the optimization subproblems of the Branch and Bound method. In
comparison to a first implementation of this Russian Dolls method from the
literature, several improvements are presented. Some of them are adaptations of
techniques already employed successfully in Branch and Bound algorithms, like
the use of approximate coloring for pruning purposes and bit-parallel
operations.
Two different coloring heuristics are tested: the standard greedy and the greedy
with recoloring. Other improvements are directly related to the Russian Dolls scheme:
the adoption of recursive calls where each subproblem (doll) is solved itself
via the same principles than the Russian Dolls Search and the application of an
elimination rule allowing not to generate a significant number of dolls.
Results of computational experiments show that the algorithm outperforms the 
best exact combinatorial algorithms in the literature for the great majority of 
the dense graphs tested, being more than twice faster in several cases.
\end{abstract}

\section{Introduction}

\subsection{Problem Statement}

Let $G=(V,E)$ be a simple and undirected graph, with $V$ being its set of vertices and $E$ its set of edges.
A {\em clique of $G$} is a subset (of $V$) of pairwise adjacent vertices.
We consider the {\em CLIQUE problem}, which consists in finding in $G$ a clique
of maximum size $\omega(G)$, which in turn is called the {\em clique number of
$G$}. In addition to its many practical applications (see for instance
\cite{Bomz99,Bute06,tomitaAkutsuMatsunaga11}), it is algorithmically equivalent
to the maximum stable set and the minimum vertex cover problems ($S \subseteq V$ is a {\em stable set}
of $G$ if it is a clique in the complement of $G$ and a {\em vertex cover} if
every edge in $E$ has at least one endpoint in $S$). The CLIQUE problem is in
NP-hard~\cite{Karp.72} and is even hard to approximate by a reasonable factor~\cite{Aror98}, unless the graph is
restricted to have a special structure. In this paper we deal with exact
algorithms for determining the clique number of arbitrary graphs.

Before going into the details of the problem and its algorithms, let us state
some notation. 
\begin{itemize}
\item $V = \{ 1, 2, \ldots, n \}$, for some $n \in \mathbb{N}$.
\item $N(u) = \{ v \in V \mid (u,v) \in E \}$ is the {\em neighborhood} of a 
vertex $u$ in $G$ whose members are {\em neighbors} of $u$.
\item If $U \subseteq V$, then $G[U] = (U,E[U])$ denotes the subgraph of $G$ 
induced by $U$. 
\item If $v \in V$, then $U+v$ and $U-v$ stands for $U \cup \{ v \}$ and 
$U \backslash v$, respectively.
\item An {\em $\ell$-coloring} of $G$ is an assignment of a color from $\{ 1,
\ldots, \ell \}$ to every vertex of $G$ such that the endpoints of any edge get different colors.
It can be characterized by $\ell$ disjoint subsets $C_1, \ldots, C_\ell$
such that $\cup_{i = 1}^{\ell} C_i = V$ and $G[C_i]$ is a stable set for all
$i \in \{ 1, \ldots, \ell \}$.
\end{itemize}

\subsection{Exact Algorithms via Branch and Bound}
\label{ssec:exactbb}

Several Branch and Bound (B\&B) algorithms have been proposed to solve the
CLIQUE problem exactly (for an overview, see~\cite{WuHao.15}). As usual, such
algorithms perform a search in a tree.
A node in this tree is a pair $(Cq, Cd)$ of disjoint subsets of $V$, where $Cq$ is
a clique of $G$ and $Cd$ is a set of candidate vertices, {\em e.g.} vertices that can
extend $Cq$ to a larger clique of $G$. In this manner, a node of the search
tree can be alternatively seen as the root of the search tree of a smaller
instance of the CLIQUE problem, more specifically the one defined on the
subgraph $G'$ of $G$ induced by $Cd$. In addition to this recursive view of the
search, some of the former algorithms employ relatively sophisticated
procedures to obtain upper bounds for $\omega(G')$ as tight as possible in the
hope of pruning the enumeration considerably~\cite{Bala86,Carr90,Mannino94}.
Since the computation of such a bound is generally applied at numerous nodes
of the search tree, the most recent developments were achieved with the use of simpler
and faster, but still effective bounding procedures.
In this vein, the most successful approach involves the use of approximate
colorings of selected subgraphs of $G$. This bound, whose application for
the CLIQUE problem was first proposed in~\cite{Fahl02}, is based on the
following remark:
\begin{rmk}[Upper bound from vertex coloring]
If $G$ admits an $\ell$-coloring, then $\omega(G) \leq \ell$.
\label{prop:upbound}
\end{rmk}
A direct consequence of this remark is that any heuristic that provides a proper
coloring of $G$ gives an upper bound for $\omega(G)$, in special the so called
{\em greedy coloring} heuristic: enumerate the elements of $V$ in some pre-defined order, assigning to
each vertex the smallest available color.

The algorithm \Call{MCR}{} proposed in~\cite{Tomi07} is very representative of
this approach. When a node of the search tree is explored,
the greedy coloring heuristic is applied considering that the corresponding
candidate vertices are stored in an array, say $R$. The order of the vertices in $R$ defines the
order of enumeration of the greedy coloring heuristic. The resulting coloring is
then used to resort $R$ in a non-decreasing order of colors. After that, the
color $c(i)$ of $R[i]$ is an upper bound for the clique number of $G[\{ R[1],
\ldots, R[i] \}]$ by Remark~\ref{prop:upbound}. Hence, vertex $R[i]$ produces a
branching only if $c(i)+$(the size of the clique defining the current node) is
greater than the best clique found so far. A branching of $R[i]$ consists in the 
generation of the node defined by the addition of $R[i]$ to the current
clique and the set of candidates $\{ R[1], \ldots, R[i-1] \} \cap N(R[i])$.
Experiments with this algorithm show that it attains a good tradeoff between time spent computing approximate colorings and number of nodes explored in the search tree.

It was shown with computational experiments that \Call{MCR}{} clearly outperformed
other existing algorithms in finding a maximum clique. However, some improvements not too time-consuming 
in comparison with the reduction in the search space thereby obtained have been
performed in this basic algorithm. In \cite{Konc07}, a more judicious ordering
of the vertices in the nodes of the search tree is proposed, improving the bounds obtained with the 
greedy coloring heuristic. In~\cite{Tomi13} (algorithm \Call{MCS}{}), the
algorithm is modified in two points related to the coloring heuristic: first, a static order similar to the 
one proposed in~\cite{Konc07} is adopted; second, a color exchange strategy is
employed to try to recolor a vertex $v$ getting a large color with a 
smaller one that could avoid the branching of $v$. Studies of the impact of vertex ordering in the coloring based
strategies mentioned above can be found in~\cite{Prosser.12,SegundoLopezBatsyn.14}. In \cite{Bats13} and \cite{Masl13},
an heuristic is applied first: the Iterated Local Search (ILS) heuristic
proposed in~\cite{Andr12} to obtain an initial high-quality solution that allows to prune early branches 
of the search tree.

\subsection{Bit-Level Parallelism}

Another improvement of the \Call{MCR}{} algorithm is accomplished by means of the
encoding of the graph as a bitmap and the incorporation of bit-parallel
operations. A leading algorithm in this direction, referred to as \Call{BBMC}{},
is described in~\cite{Segu11,Segu13}. A {\em bitmap} is a data structure for set
encoding which stores individual elements of the set in a compact form while allows for direct
address of each element.
Still more interesting is its ability to benefit from the
potential bit-level parallelism available in hardware to perform collective set
operations through fast bit-masking operations (intersection of two
sets is a typical example detailed in Section~\ref{sec:bitpar}).
However, to exploit this potential parallelism in practice to improve
overall efficiency is not a trivial task since the manipulation of bitmaps
turns out to be less efficient when the enumeration of elements is 
relevant~\cite{Segu08}.

The use of bit-masking operations occurs in the \Call{BBMC}{} algorithm in two
points, namely the branching and the greedy coloring heuristic. Branching corresponds to
determining the set intersection $\{ R[1], \ldots, R[i-1] \} \cap N(R[i])$,
whereas a greedy coloring can be built as successive operations of set difference of the neigborhood of selected vertices
with a set of candidates. In this sense, set intersection and set difference are the
essential operations in the \Call{BBMC}{} algorithm. These are set operations that
are efficiently performed by means of bit-masking operations if the sets
involved are stored as bitmaps. For this reason, bit-level parallelism has been
proved to be a powerful tool in efficient implementations of the branching and
bounding rules of the \Call{BBMC}{} algorithm. Naturally, this
requires that the input graph and the nodes of the search tree are stored
as bitmaps.

A modification of the \Call{BBMC}{} to encompass the recoloring strategy of the
\Call{MCS}{} algorithm and to further reduce the time spent in the coloring based bounding
procedure are the contributions in~\cite{Segu14}. The reduction in the
coloring computation time is accomplished as follows. Let us consider $(Cq, Cd)$
as the current node and $MAX$ as the size of the best known solution found so
far. The improved bounding procedure consists in determining a maximal subgraph
of $G[Cd]$ that is $k$-partite, for $k = MAX-|Cq|$, with a partial coloring
greedy heuristic. The vertices so colored are not considered for branching.
Specific rules are used to sort and select uncolored vertices for branching. Experimental results show that the
number of nodes visited is almost always greater than \Call{BBMC}{}, but
significant improvement in performance occurs for a certain number of graphs.

\subsection{An Exact Algorithm via the Russian Dolls Method}
\label{sec:rdalg}

An alternative method, called {\em Russian Dolls (RD)} in its original
description in~\cite{Verf96}, has also been used to solve the CLIQUE problem. 
When compared to B\&B, the main difference of the RD
method is that the nodes of its search tree correspond to decision subproblems,
instead of the optimization subproblems of the B\&B method. In
general terms, the method consists in iteratively solving larger and larger subproblems (also referred to as {\em dolls}) to optimality until the global problem is
solved. During this iterative process, the optimum value of each doll is taken
into account when solving larger subproblems. As far as we are aware, the
original application of this method to the CLIQUE problem is the algorithm proposed in~\cite{Oste02}, in which subproblems are
associated with subgraphs $G_i = (V_i, E_i)$, $i \in \{1,...,n\}$, where $V_1 =
\{1\}$, $V_{i+1} = V_i \cup \{i+1\}$, $E_i = E[V_i]$, and $G_n = G$ (for the
sake of convenience, we use a slight modification of the notation used
in~\cite{Oste02}). An optimum solution of the doll of index $i$ is a clique of
maximum size in the associated subgraph $G_i$, which means that $\omega(G_i)$ is known after solving doll of index $i$.
Thus, searching for a maximum clique in $G_{i+1}$ corresponds to decide whether
$\omega(G_{i+1}) = \omega(G_i)$ or $\omega(G_{i+1}) = \omega(G_i) + 1$. Moreover, $\omega(G_{i+1})$ can be equal to $\omega(G_i) + 1$ only if the unique
vertex $i+1$ in $V_{i+1} \setminus V_i$ appears in every maximum clique of
$G_{i+1}$. For this reason, doll of index $i+1$ is handled only once $G_i$ is
solved by solving the {\em decision subproblem} of deciding whether $G[V_i \cap
N(i+1)]$ contains a clique of size $\omega(G_i)$ or not. Every decision
subproblem is an instance of an NP-Complete problem, but hopefully of small or moderate
size.

In the absence of effective strategies to reduce the search space, the time
required to solve ``no'' decision subproblems can become prohibitively high,
even for moderately sized instances. In order to try to circumvent this
drawback, there are two pruning rules devised in~\cite{Oste02} to cut a ``no'' decision subproblem associated with the set $V_i \cap N(i+1)$ of candidates, as follows:
\begin{enumerate}
\item $|V_i \cap N(i+1)| < \omega(G_i)$: in this situation, $V_i \cap
N(i+1)$ does not contain enough candidates to build a clique of size
$\omega(G_i)$; and
\label{it:smallinst}
\item there exists no $j \in V_i \cap N(i+1)$ such that $\omega(G_j) =
\omega(G_i)$: this is equivalent to say that $\omega(G_j)$, for all $j \in
V_i$, is such that $\omega(G_j) < \omega(G_i)$ (recall that $\omega(G_j)$ has
been already computed). Since $\omega(G_j)$ is an upper bound for $\omega(G[V_j
\cap N(i+1)])$, we can conclude that no clique in $G[V_i \cap N(i+1)]$ has size $\omega(G_i)$.
\label{it:rdprune}
\newcounter{prunrules}
\setcounter{prunrules}{\value{enumi}}
\end{enumerate}
It is worth remarking that rule~\ref{it:rdprune} does not imply
rule~\ref{it:smallinst}. The effectiveness of rule~\ref{it:rdprune} on pruning
``no'' decision subproblems depends on how $V_j \cap N(i+1)$ differs from
$V_j$. More specifically, when there exists $j \in V_i \cap N(i+1)$ such
that $\omega(G_j) = \omega(G_i)$ and $\omega(G[V_j \cap N(i+1)]) < \omega(G_j)$,
the pruning rule 2 fails to prune $G_j$. In this sense, a contribution of the
algorithm proposed in this paper with respect to the one in~\cite{Oste02} is the use of more effective bounding heuristics.

\subsection{Our Contributions}
\label{sec:contributions}

In this paper, we propose a new exact algorithm for the CLIQUE problem based on
the RD method. The goal is to provide the basic algorithm
in~\cite{Oste02} with an alternative pruning rule which allows to skip a larger
number of dolls. As a result, the number of ``no'' decision subproblems examined is significantly reduced. For this purpose,
we incorporate procedures that have already shown their effectiveness with the
B\&B method but, as far as we know, their performance have not yet been
checked with a RD framework. These procedures need to be adapted having in mind
that simplicity is very important to make the computation overhead as low as
possible. In this sense, we suggest the following improvements to the original
implementation of the RD algorithm presented in~\cite{Oste02}:
\begin{itemize}
\item The use of partial coloring heuristics to establish the sequence of
dolls $\langle G_1 = (V_1, E_1), \ldots, G_n = (V_n, E_n) \rangle$. In our
algorithm, contrary to the one in~\cite{Oste02}, the sequence $\langle v_1,
\ldots, v_n \rangle$ of vertices defining $V_1 = \{v_1\}$ and $V_i = V_{i-1} \cup \{v_i\}$, for $i
\in \{ 2, \ldots, n \}$, is not determined beforehand.
Instead, the order in which the vertices are considered is established during the execution of the algorithm
in order to eliminate as many dolls as possible. For this purpose, once a
decision subproblem $G_i$, $i \in \{ 1, \ldots, n-1 \}$, is solved, the choice
of the next doll to handle is made depending on the answer of $G_i$. If $G_i$ is
a ``no" instance, then we choose $v_{i+1}$ as the smallest vertex in $V
\setminus V_i$. Otherwise, $G_i$ is a ``yes" instance, which means
that the current best solution is incremented as a result of solving $G_i$.
Thus, we apply the following elimination rule, based on
Remark~\ref{prop:upbound}. Let $C$ be the clique of size $\omega(G_i)$ found in
$G_i$. We first apply a greedy heuristic to extend $C$ to a maximal clique $C'$
of $G$ by adding $k = (|C'|-|C|+1)$ vertices from $V \setminus V_i$.
Then, we search for a maximal $k$-partite induced subgraph of $G[V \setminus
V_i]$. Let us say that $L$ is the set of vertices found in this search, with
$|L|=\ell$. We set $V_{i+\ell} = V_i \cup L$ and choose $v_{i+\ell+1}$ as the smallest vertex in
$V \setminus V_{i+\ell}$. This corresponds to eliminate the decision subproblems
$G_{i+1}, \ldots, G_{i+\ell}$. We tested variations of the coloring heuristics
used in \Call{MCR}{} and \Call{MCS}{} to find $k$-partite induced subgraphs,
namely:
the ``standard greedy'' \cite{Tomi03} and the greedy with recoloring
\cite{Tomi13}.
\item Each decision subproblem
$G_{i+1}$ is itself solved through a recursive enumeration based on the same
principles of the RD method. More precisely, the search is performed in larger
and larger subdolls of $G[V_i \cap N(i+1)]$ until either a clique of size
$\omega(G_i)$ is found or it is proved that no such clique exists. However, a
particularity of this enumeration is that it follows its own sequence of
subdolls in the sense that the associated sequence of vertices is not
(necessarily) a subsequence of $\langle v_1, \ldots, v_i \rangle$. The reason is
that the pruning rule described in the previous item is applied as an initial
step to determine (and prune) an $(\omega(G_i)-1)$-partite induced subgraph $H$
of $G[V_i \cap N(i+1)]$. The sequence $\langle w_1, \ldots, w_{|V_i \cap
N(v_{i+1})|} \rangle$ of subdolls is such that $V(H) = \{ w_1, \ldots,
w_{|V(H)|} \}$ and, for every $j > |V(H)|$, $w_j$ is the smallest vertex of
$(V_i \cap N(i+1)) \setminus \{ w_1, \ldots, w_{j-1} \}$. As a consequence, the
coloring based pruning rule enhances pruning rules~\ref{it:smallinst}
and~\ref{it:rdprune}. We give more details of this fact in
Section~\ref{sec:overall}.
\item A bitmap encoding of $G$ and an optimized implementation of several 
procedures to benefit from 128-bit parallelism available in the Streaming SIMD
Extensions CPU instruction set. This requires the reformulation of the RD method
as an appropriate sequence of set operations. When compared with the B\&B implementations described
in~\cite{Segu11,Segu13,Segu14}, we adopt the same principle of using bitmaps, but we
describe the sequence of set operations in order not to need to use arrays of integers. Consequently, our
implementation requires less memory space for these data structures.
\end{itemize}

Extensive computational experiments have been carried out to compare our
algorithm, called \Call{RDMC}{}, with effective algorithms from the literature.
The algorithms chosen for comparison purposes were \Call{MCR}{}, \Call{MCS}{},
and recent versions of \Call{BBMC}{} because it has been showed in previous
works that they clearly outperform the algorithm in~\cite{Oste02}. We use in our analyses a specific implementation of
\Call{MCR}{}, \Call{MCS}{}, and \Call{BBMC}{} inspired in~\cite{Segu14}. In
particular, we add a slight improvement that contributes to further reduce the
number of explored subproblems. The same basic routines for set operations
are used in all our implementations. All the computational experiments were ran
in the same computational platform. This aims at avoiding the imprecision
resulting from experiments with distinct codes or computational platforms, as
pointed out in~\cite{Prosser.12}.
Results show that our implementations outperform the best exact combinatorial
algorithms in the literature. In addition, our implementation of \Call{RDMC}{}
is often more efficient than the B\&B counterparts for graphs with density above
80\%, being more than twice faster in several cases.

The remainder sections are organized as follows. In Section~\ref{sec:overall},
we give an overall description of our RD algorithm, describing
its main elements. The details of the algorithm is the subject of Section~\ref{sec:newfeat}. This section includes the description of the different improvements and specific features listed above. 
Finally, experimental results and analyses are presented in 
Section~\ref{sec:exp}. The paper is closed with
some concluding remarks in Section~\ref{sec:conc}.

\section{Overall Description of the Algorithm}
\label{sec:overall}

In this section, we give a general overview of our RD algorithm. The
main elements of this algorithm, which are outlined in
Alg.~\ref{alg:genmethod}, are the iterative procedure for decision subproblem
generation and applications of the pruning rule. In what follows, we describe
these main elements. We postpone the details on how these elements are implemented until Section~\ref{sec:bitpar}.

\begin{algorithm}[h]
\caption{RD for the CLIQUE problem}
\label{alg:genmethod}
\begin{algorithmic}[1]
{\small
\State $MAX \gets 0$, $R \gets \emptyset$, $S \gets V$
\While {$S \ne \emptyset$} \label{lin:whileSbegin}
	\State Let $v$ be the smallest vertex in $S$ \label{lin:selv}
	\State $Cq \gets \emptyset$
	\If {\Call{decide}{$G$, $R \cap N(v)$, $MAX$, $Cq$}} \label{lin:calldecide}
		\State $k \gets$ \Call{extendClique}{$G$, $S$, $Cq$} + 1 \label{lin:extCq}
		\State \Call{maxPartiteSubgraph}{$G$, $S$, $R$, $k$} \label{lin:maxPart}
		\State $MAX \gets MAX + k$
	\Else
		\State $S \gets S - v$,  $R \gets R + v$ \label{lin:whileSend}
	\EndIf
\EndWhile
}
\end{algorithmic}
\end{algorithm}

\subsection{Decision Subproblems Enumeration} 
\label{sec:subproblems}

Recall that the general description of the RD method establishes
that, for every $i \in \{ 1, \ldots, n \}$, $V_i$ stands for the subset of
vertices defining a subgraph $G_i$ and, thus, a decision subproblem. In
Alg.~\ref{alg:genmethod}, variables $R$ and $S$ are used to store the current
decision subproblem defining set $V_i$ and its complement $V \setminus V_i$,
respectively. Their initial states correspond to an empty decision
subproblem. Another set variable, $Cq$, is used to store the current clique of
$G$, whereas the integer variable $MAX$ contains the size of the maximum clique
found so far. The enumeration is performed iteratively in the {\bf while} loop
between lines~\ref{lin:whileSbegin}--\ref{lin:whileSend}. The first step in an iteration
is the choice, at line~\ref{lin:selv}, of the vertex $v$ (which is called $v_{i+1}$ in the general description of the method) to be moved
from $S$ to $R$ in order to define the new current doll.
It is worth remarking that the choice of $v$ determines the enumeration order. It follows that the choice of $v$ as the smallest vertex in $R$ depends not only on the vertex numbering but
also on the result of the application of the pruning rule in previous
iterations. After executing line~\ref{lin:selv}, we have to decide (with the
recursive function \Call{decide}{}) whether the decision subproblem $G[R \cap
N(v)]$ has a clique of size $MAX = \omega(G[R])$. If this search fails, then
$\omega(G[R + v]) = \omega(G[R])$ and we go to the next iteration. Otherwise,
there exists a clique of $G[R+v]$ containing $v$ that gives
$\omega(G[R+v]) = \omega(G[R])+1 = MAX+1$, which enables the
application of the pruning rule.
Function \Call{decide}{} gets four parameters as input,
namely the graph $G$, a subset $R$ of candidate vertices, an integer 
$\ell$, and an empty clique $Cq$. It returns TRUE if and only if $G[R]$ contains
a clique of size $\ell$, in which case $Cq$ contains such a clique. Otherwise, it
returns FALSE.

\subsection{Enhanced Pruning Rule}
\label{sec:prule}

A decision subproblem can be pruned from the search when it can be proved in the
function call \Call{decide}{$G$, $R$, $\ell$, $Cq$} that it does not contain 
any clique of the desired size $\ell$. The enhanced pruning rule of our RD algorithm is applied at lines~\ref{lin:extCq} and~\ref{lin:maxPart}. For
line~\ref{lin:extCq}, it can be noticed that the clique $Cq$ of $G[R \cap
N(v)]$ obtained when the call \Call{decide}{$G$, $R \cap N(v)$, $MAX$, $Cq$}
returns TRUE is not necessarily maximal in $G$. So, we make it maximal with
call \Call{extendClique}{$G$, $S$, $Cq$}, which returns the number $k$ of
vertices of $S$ added to $Cq$. Since we have $MAX + k$ as the new lower bound
for the optimum solution, we can prune subproblems corresponding to vertices in
$S$ based on a generalization of Remark~\ref{prop:upbound}. The basic property
used in this pruning is that, in a stable set, at most one vertex can belong to
the maximum clique: indeed two vertices that are not linked cannot belong
conjointly to a clique. More fundamentally, we use the following property:

\begin{prop}
Let $R, R' \subseteq V$, $R \subseteq R'$ be such that $G[R' \setminus R]$
admits an $k$-coloring, $k \geq 1$. Then, $\omega(R') \leq \omega(R)+k$.
\label{prop:maxclprun}
\end{prop}

\begin{proof}
Since $R' = R \cup (R' \setminus R)$, we have $\omega(G[R']) \leq 
\omega(G[R]) + \omega(G[R' \setminus R])$. In addition, since at most one
vertex of each color can be in a clique, $\omega(G[R' \setminus R]) \leq k$
and the result follows.
\nqed
\end{proof}

A consequence of Property~\ref{prop:maxclprun} is that $R'$ is built from
$R$ by the addition of vertices from $S$ defining a $k$-partite subgraph of $G$.
Then, $\omega(G[R'])$ is at most the new lower bound $MAX + k$. For this
purpose, we make call \Call{maxPartiteSubgraph}{$G$, $S$, $R$, $k$}, which moves
a maximal $k$-partite subgraph of $G[S]$ from $S$ to $R$. Note that the
vertices that have been added to $Cq$ by the previous call to function
\Call{extendClique}{} do not need to be moved from $S$ to $R$. Note also that
the iterations corresponding to the vertices removed from $R$ are skipped from
the enumeration.

A particular application of this pruning rule occurs at the first iteration of
Alg.~\ref{alg:genmethod}. In this case, \Call{decide}{$G$, $\emptyset$, 0,
$\emptyset$} returns TRUE. So, the algorithm begins by determining a maximal
clique of $G$, of size $k$, which leads the call \Call{maxPartiteSubgraph}{$G$,
$S$, $R$, $k$} to provide a maximal $k$-partite subgraph of $G$, of cardinality,
say, $\ell$. The corresponding $\ell$ iterations are skipped and the RD process
starts on $G_{\ell+1}$.

\subsection{Sequence of Dolls}

The sequence $\langle v_1, \ldots, v_n \rangle$ of vertices defining $V_1 =
\{v_1\}$ and $V_i = V_{i-1} \cup \{v_i\}$, for $i \in \{ 2, \ldots, n \}$, and,
consequently, the sequence of dolls $\langle G_1 = (V_1, E_1), \ldots, G_n = (V_n, E_n) \rangle$,
is determined during the execution of Alg.~\ref{alg:genmethod} in the following
way. Due to the pruning rule, the vertices $v_1, \ldots, v_r$ are of two types
with respect to the way they have entered $R$: there are those moved from $S$
(i) at line~\ref{lin:whileSend} and (ii) by a call to
function \Call{maxPartiteSubgraph}{} at line~\ref{lin:maxPart}. Let $it(v)$, $v
\in V$, be the iteration in which $v$ is inserted in
$R$ (we assume that the first iteration has rank 1). Then, for every $i, j
\in \{ 1, \ldots, n \}$, $i < j$, the following conditions hold:
\begin{enumerate}
  \item $it(v_i) < it(v_j)$, or
  \item $it(v_i) = it(v_j)$ (this means that both $v_i$ and $v_j$ are of type
  (ii)) and there exists an ordering $\langle C_1, \ldots, C_k \rangle$ of the
  $k$ stable sets determined in the call to \Call{maxPartiteSubgraph}{} at
  iteration $it(v_i)$ such that $v_i \in C_{c(v_i)}$ and $v_j \in C_{c(v_j)}$
  yields $c(v_i) \leq c(v_j)$.
\end{enumerate}

Besides determining the sequence of dolls, the application of the enhanced
pruning rule has the effect that the only information available at any iteration of
Alg.~\ref{alg:genmethod} about some previous dolls is an upper bound for their
optimum, and not their exact value. To make it more precise, let $m_i=0$, if
$i=0$, and, for $i \in \{ 1, \ldots, n \}$, let $m_i$ be $m_{i-1}$ if $v_i$ is
of type (i), or $c(v_i)$ plus the value of $MAX$ at the begining of iteration $it(v_i)$
if $v_i$ is of type (ii). Clearly, $\langle m_1, \ldots, m_r \rangle$ is a
nondecreasing sequence and $m_i$ is equal to $\omega(G_i)$ if $v_i$ is of type
(i) or an upper bound for $\omega(G_i)$, if $v_i$ is of type (ii). In
addition, the dolls corresponding to vertices of type (i) have the following property.

\begin{prop}
Let $v$ be the vertex selected at line~\ref{lin:selv} at the current iteration
of Alg.~\ref{alg:genmethod}. Then, for all $b \in \{ 1, \ldots, MAX \}$, there
exists $v_t \in R \cap N(v)$ such that $m_t = b$.
\label{prop:mi}
\end{prop}

\begin{proof}
First note that, by definition, vertex $v_j \in R$ having $m_j = b$ with
smallest index $j$ is of type (ii), which means that $v_j$ is included in the
stabe set $C_{c(v_j)}$ at iteration $it(v_j)$. Since $C_{c(v_j)}$ is maximal
with respect to $S \cup C_{c(v_j)}$ and $C_{c(v_j)} \subseteq R$, the set $R
\cap N(v)$ contains a neighbor $v_t$ of $v$ with $m_t=b$.
\nqed
\end{proof}

A final remark with respect to the sequence of dolls established by
Alg.~\ref{alg:genmethod} is that an adaptation of pruning rule~\ref{it:rdprune}
to be used with $m_j$ in the role of $\omega(G_j)$ is useless.
Property~\ref{prop:mi} with $b = \omega(G_i)$ implies that the adapted
pruning rule~\ref{it:rdprune} would not eliminate the decision subproblem $R
\cap N(v)$.


\section{A Comparative Example}

\tikzstyle{vertex}=[circle,minimum size=8pt]

The main differences between RD and B\&B algorithms, both
using partial colorings for pruning purposes, are illustrated in this section
by means of an example. Indeed, as one may see with this example, the
differences between the two approaches make them complementary to each other. We show the executions of
Alg.~\ref{alg:genmethod} and an improved version of the algorithm in~\cite{Segu14} for the graph of Fig.~\ref{fig:graph}. In both cases, we consider that the same greedy heuristics for clique and
partial coloring generation are used. In these heuristics, vertices are examined
in an increasing order of vertex identity.

\begin{figure}[htbp]
\centering
		\begin{tikzpicture}[thick,fill opacity=0.9,scale=0.8,node distance=0.5cm]
		\input{graphExample}
		\end{tikzpicture}
\caption{Graph}
\label{fig:graph} 
\end{figure}

In Fig.~\ref{fig:execrd}, the iterations of Alg.~\ref{alg:genmethod} are
represented by the decision subproblems generated and the states of the two sets
$R$ and $S$ attained during its manipulations. In the first iteration, the
corresponding decision subproblem is related to the empty subgraph of $G$ and
$MAX=0$. Since it is a ``yes'' instance, an extended clique (starting with clique $\{ 0 \}$) and a partial coloring are construted, generating
the states depicted in Fig.~\ref{sfig:it1}. Extending $\{ 0 \}$ in an increasing
order of vertex indices gives the maximal clique $\{ 0, 1 \}$. Thus, $MAX$ is
incremented by 2, which leads to a partial coloring of the
vertices in $S$ with the 2 colors $C_1 = \{0,2,5\}$ and $C_2 = \{1,3,9\}$,
pruning the so colored vertices by moving them to $R$. It is straightforward to check that, as predicted by Property~\ref{prop:upbound},
the subgraph induced by the colored vertices does not contain any clique of size
larger than 2. In the second iteration (Fig.~\ref{sfig:it2}), when selected, the
smallest vertex in $S$ (vertex 4) generates a ``no'' decision subproblem. Note
that, according to Property~\ref{prop:mi}, this decision subproblem cannot be
eliminated because $m_5 = 2$, for $v_5 = 3$. Finally, in the third iteration of
Fig.~\ref{sfig:it3}, $MAX$ is incremented by 2 again, which prunes all remaining
vertices in $S$. In summary, the execution has the characteristics shown in
Table~\ref{tab:rd}.

\begin{figure}[htbp]
\centering {\scriptsize
		\subfigure[First iteration. Vertex selected is 0. Two states attained after
		call {\sc decide}$(G, \emptyset, 0, \emptyset)$. At the end,
		$MAX = 2$.]{\label{sfig:it1}
		\begin{tabu}{r|[0.7pt]c|[0.7pt]c|[0.7pt]c|[0.7pt]c|[0.7pt]c|[0.7pt]c|[0.7pt]c|[0.7pt]c|[0.7pt]c|[0.7pt]c|[0.7pt]}
		\tabucline[0.7pt]{2-11} 
		$S =$ & 0 & 1 & 2 & 3 & 4 & 5 & 6 & 7 & 8 & 9 \\ \tabucline[0.7pt]{2-11} 
		$R =$ & \cellcolor[gray]{.8} & \cellcolor[gray]{.8} & \cellcolor[gray]{.8} &
		\cellcolor[gray]{.8} & \cellcolor[gray]{.8} & \cellcolor[gray]{.8} &
		\cellcolor[gray]{.8} & \cellcolor[gray]{.8} & \cellcolor[gray]{.8} &
		\cellcolor[gray]{.8} \\ \tabucline[0.7pt]{2-11} 
		$Cq=$ & 0 & 1 & \cellcolor[gray]{.8} &
		\cellcolor[gray]{.8} & \cellcolor[gray]{.8} & \cellcolor[gray]{.8} &
		\cellcolor[gray]{.8} & \cellcolor[gray]{.8} & \cellcolor[gray]{.8} &
		\cellcolor[gray]{.8} \\ \tabucline[0.7pt]{2-11}
		\multicolumn{11}{c}{After \Call{extendClique}{$G, \{ 0\text{--}9 \}, \emptyset$}} \\
		\multicolumn{11}{c}{}  \\ \tabucline[0.7pt]{2-11}
		$S =$ & \cellcolor[gray]{.8} & \cellcolor[gray]{.8} & \cellcolor[gray]{.8} &
		\cellcolor[gray]{.8} & 4 & \cellcolor[gray]{.8} &
		6 & 7 & 8 &
		\cellcolor[gray]{.8} \\ \tabucline[0.7pt]{2-11} 
		$R =$ & \cellcolor{green!60} 0 & \cellcolor{blue!60} 1 & \cellcolor{green!60} 2 & 
		\cellcolor{blue!60} 3 & \cellcolor[gray]{.8} & \cellcolor{green!60} 5 &
		\cellcolor[gray]{.8} & \cellcolor[gray]{.8} & \cellcolor[gray]{.8} &
		\cellcolor{blue!60} 9 \\ \tabucline[0.7pt]{2-11} 
		$Cq=$ & 0 & 1 & \cellcolor[gray]{.8} &
		\cellcolor[gray]{.8} & \cellcolor[gray]{.8} & \cellcolor[gray]{.8} &
		\cellcolor[gray]{.8} & \cellcolor[gray]{.8} & \cellcolor[gray]{.8} &
		\cellcolor[gray]{.8} \\ \tabucline[0.7pt]{2-11}
		\multicolumn{11}{c}{After \Call{maxPartiteSubgraph}{$G, \{ 0\text{--}9 \}, \emptyset, 2$}} \\
		\multicolumn{11}{c}{}  \\ 
		\end{tabu}}\quad\quad
		\subfigure[Second iteration. Vertex selected is 4. A state attained after
		call {\sc decide}$(G, \{ 0, 3, 5 \}, 2, \emptyset)$.]{\label{sfig:it2}
		\begin{tabu}{r|[0.7pt]c|[0.7pt]c|[0.7pt]c|[0.7pt]c|[0.7pt]c|[0.7pt]c|[0.7pt]c|[0.7pt]c|[0.7pt]c|[0.7pt]c|[0.7pt]}
		\tabucline[0.7pt]{2-11} 
		$S =$ & \cellcolor[gray]{.8} & \cellcolor[gray]{.8} & \cellcolor[gray]{.8} &
		\cellcolor[gray]{.8} & \cellcolor[gray]{.8} & \cellcolor[gray]{.8} &
		6 & 7 & 8 &
		\cellcolor[gray]{.8} \\ \tabucline[0.7pt]{2-11} 
		$R =$ & \cellcolor{green!60} 0 & \cellcolor{blue!60} 1 & \cellcolor{green!60} 2 & 
		\cellcolor{blue!60} 3 & 4 & \cellcolor{green!60} 5 &
		\cellcolor[gray]{.8} & \cellcolor[gray]{.8} & \cellcolor[gray]{.8} &
		\cellcolor{blue!60} 9 \\ \tabucline[0.7pt]{2-11} 
		$Cq=$ & 0 & \cellcolor[gray]{.8} & \cellcolor[gray]{.8} &
		\cellcolor[gray]{.8} & \cellcolor[gray]{.8} & \cellcolor[gray]{.8} &
		\cellcolor[gray]{.8} & \cellcolor[gray]{.8} & \cellcolor[gray]{.8} &
		\cellcolor[gray]{.8} \\ \tabucline[0.7pt]{2-11}
		\multicolumn{11}{c}{}  \\ 
		\end{tabu}}\quad\quad
		\subfigure[Third iteration. Vertex selected is 6. Two states attained after
		call {\sc decide}$(G, \{ 2, 3, 5 \}, 2, \emptyset)$. At the end,
		$MAX = 4$.]{\label{sfig:it3}
		\begin{tabu}{r|[0.7pt]c|[0.7pt]c|[0.7pt]c|[0.7pt]c|[0.7pt]c|[0.7pt]c|[0.7pt]c|[0.7pt]c|[0.7pt]c|[0.7pt]c|[0.7pt]}
		\tabucline[0.7pt]{2-11} 
		$S =$ & \cellcolor[gray]{.8} & \cellcolor[gray]{.8} & \cellcolor[gray]{.8} &
		\cellcolor[gray]{.8} & \cellcolor[gray]{.8} & \cellcolor[gray]{.8} &
		6 & 7 & 8 &
		\cellcolor[gray]{.8} \\ \tabucline[0.7pt]{2-11} 
		$R =$ & \cellcolor{green!60} 0 & \cellcolor{blue!60} 1 & \cellcolor{green!60} 2 & 
		\cellcolor{blue!60} 3 & 4 & \cellcolor{green!60} 5 &
		\cellcolor[gray]{.8} & \cellcolor[gray]{.8} & \cellcolor[gray]{.8} &
		\cellcolor{blue!60} 9 \\ \tabucline[0.7pt]{2-11} 
		$Cq=$ & \cellcolor[gray]{.8} & \cellcolor[gray]{.8} & 2 &
		3 & \cellcolor[gray]{.8} & \cellcolor[gray]{.8} &
		6 & 7 & \cellcolor[gray]{.8} &
		\cellcolor[gray]{.8} \\ \tabucline[0.7pt]{2-11}
		\multicolumn{11}{c}{After \Call{extendClique}{$G, \{ 6, 7, 8 \}, \{ 2, 3 \}$}} \\
		\multicolumn{11}{c}{}  \\ \tabucline[0.7pt]{2-11}
		$S =$ & \cellcolor[gray]{.8} & \cellcolor[gray]{.8} & \cellcolor[gray]{.8} &
		\cellcolor[gray]{.8} & \cellcolor[gray]{.8} & \cellcolor[gray]{.8} &
		\cellcolor[gray]{.8} & \cellcolor[gray]{.8} & \cellcolor[gray]{.8} &
		\cellcolor[gray]{.8} \\ \tabucline[0.7pt]{2-11} 
		$R =$ & \cellcolor{green!60} 0 & \cellcolor{blue!60} 1 & \cellcolor{green!60} 2 & 
		\cellcolor{blue!60} 3 & 4 & \cellcolor{green!60} 5 &
		\cellcolor{red!60} 6 & \cellcolor{magenta!60} 7 & \cellcolor{red!60} 8 &
		\cellcolor{blue!60} 9 \\ \tabucline[0.7pt]{2-11} 
		$Cq=$ & \cellcolor[gray]{.8} & \cellcolor[gray]{.8} & 2 &
		3 & \cellcolor[gray]{.8} & \cellcolor[gray]{.8} &
		6 & 7 & \cellcolor[gray]{.8} &
		\cellcolor[gray]{.8} \\ \tabucline[0.7pt]{2-11}
		\multicolumn{11}{c}{After \Call{maxPartiteSubgraph}{$G, \{ 6, 7, 8 \}, \{ 0\text{--}5, 9 \}, 2$}} \\
		\multicolumn{11}{c}{}  \\ 
		\end{tabu}}
}
\caption{Execution of Alg.~\ref{alg:genmethod} for the graph of
Fig.~\ref{fig:graph}, with initial state $S = \{ 0, \ldots, 9 \}$, $R =
\emptyset$, and $MAX = 0$. The sequence of dolls corresponds to $\langle 0, 2,
5, 1, 3, 9, 4, 6, 8, 7 \rangle$.}
\label{fig:execrd}
\end{figure}

\begin{table}[htbp]
\begin{center}
  \begin{tabular}{|c|c|c|c|} \hline
   & \multicolumn{2}{|c|}{New solution} &  \\
  \multirow{-2}{*}{Decision subprolem solved} & found & extended &
  \multirow{-2}{*}{$MAX$ incremented by} \\ \hline \hline 
  $\emptyset = \emptyset \cap N(0)$ & $\{ 0, 1
  \}$ & $\{ 0, 1 \}$ & 2 \\
  $\{ 0, 3, 5 \} = \{ 0, 1, 2, 3, 5, 9 \} \cap N(4)$ & -- & -- & -- \\
  $\{ 0, 2, 3, 5 \} = \{ 0, 1, 2, 3, 4, 5, 9 \} \cap N(6)$ & $\{ 2, 3, 6 \}$ &
  $\{ 2, 3, 6, 7 \}$ & 2 \\ \hline
  \multicolumn{4}{|l|}{Colors constructed by calls to
  \Call{maxPartiteSubgraph}{}: $\{ 0, 2, 5 \}$, $\{ 1, 3, 9 \}$, $\{ 6, 8 \}$, and $\{ 7 \}$} \\ 
  \multicolumn{4}{|l|}{Uncolored vertex: $4$} \\ \hline
  \end{tabular}
\end{center}
\caption{Summary of execution at Fig.~\ref{fig:execrd}.}
\label{tab:rd}
\end{table}

The version of the B\&B algorithm in~\cite{Segu14} that we use
to compare with Alg.~\ref{alg:genmethod} is outlined in
Alg.~\ref{alg:genmethodbb}. In this algorithm, $R$ is the set of colored
vertices, $S$ is the set of uncolored candidate vertices, and $CUR$ is the
size of the clique defining the current node of the search tree. The point to
be highlighted is the call \Call{optimize}{$G$, $(R \cup S) \cap N(v)$, $CUR+1$}
at line~\ref{lin:extCqbb}. In comparison with call \Call{decide}{$G$, $R \cap
N(v)$, $MAX$, $Cq$} at line~\ref{lin:calldecide} of Alg.~\ref{alg:genmethod},
we can observe two differences. First, contrary to the RD algorithm, the
subgraph involved in the B\&B version includes the vertices in $R \cap N(v)$.
As a consequence, the second difference is that the subproblem is an
optimization problem. In Fig.~\ref{fig:execbb}, the execution of the call
\Call{optimize}{$G$, $V$, $0$} is shown. The main idea behind the recursive
function \Call{optimize}{} is the same as the \Call{BBMC}{} algorithm. With
respect to the algorithm described in Subsection~\ref{ssec:exactbb}, there is a modification originally proposed in~\cite{Segu14}. It consists in the use of a partial coloring of the set of unexplored vertices, as indicated in line~\ref{lin:maxPartbb1} of Alg.~\ref{alg:genmethodbb}, with the purpose of determining the vertices that can be pruned from the search due to their upper bounds. For the vertices that remain uncolored after line~\ref{lin:maxPartbb1}, the corresponding optimization subproblems are generated and solved recursively. As an additional
improvement, we introduce line~\ref{lin:maxPartbb2} to prune additional vertices
whenever the current best solution value is incremented during the recursive
call.

\begin{algorithm}[htbp]
\caption{Partial coloring B\&B}
\label{alg:genmethodbb}
\begin{algorithmic}[1]
{\small
\Function{optimize}{$G$, $S$, $CUR$}
\State $R \gets \emptyset$, $MAX \gets \max \{ MAX, CUR \}$, $k \gets MAX-CUR$
\State \Call{maxPartiteSubgraph}{$G$, $S$, $R$, $k$} \label{lin:maxPartbb1}
\While {$S \ne \emptyset$} \label{lin:whileSbeginbb}
	\State Let $v$ be the greatest vertex in $S$ \label{lin:selvbb}
	\State \Call{optimize}{$G$, $(R \cup S) \cap N(v)$, $CUR+1$} \label{lin:extCqbb}
	\State $S \gets S - v$ \label{lin:whileSendbb}
	\State \Call{maxPartiteSubgraph}{$G$, $S$, $R$, $MAX-CUR-k$} \label{lin:maxPartbb2} 
	\State $k \gets MAX-CUR$
\EndWhile
\EndFunction
}
\end{algorithmic}
\end{algorithm}

The execution of the partial coloring B\&B algorithm results in
the generation of three optimization subproblems, corresponding to the
three recursive calls represented in figures~\ref{sfig:it1bb},~\ref{sfig:it2bb},
and~\ref{sfig:it3bb}. The execution is summarized in Table~\ref{tab:bb}.

There are some relevant remarks with respect to the subproblems generated during
the execution of Alg.~\ref{alg:genmethod} and Alg.~\ref{alg:genmethodbb}.
First, the optimization subproblems tend to be defined by larger subgraphs than
the decision subproblems as a consequence of the complementary selection
strategies (lines~\ref{lin:selv} of Alg.~\ref{alg:genmethod} and
line~\ref{lin:selvbb} of Alg.~\ref{alg:genmethodbb}). Second, the upper bound
used to prune nodes tends to be tighter with the selection strategy of
Alg.~\ref{alg:genmethod}. For instance, the upper bound for the subgraph induced
by $\{ 0, 1, 2, 3, 4, 5 \}$ obtained in the second iteration of Alg.~\ref{alg:genmethod} (Fig.~\ref{sfig:it2}) is 2, whereas the upper bound for the same subgraph is 3 in second optimization subproblem
(Fig.~\ref{sfig:it2bb}). On the other hand, the choice
of the greatest vertex in line~\ref{lin:selvbb} of Alg.~\ref{alg:genmethodbb},
complementary to the choice of the smallest one in Alg.~\ref{alg:genmethod},
tends to generate better lower bounds faster. One example occurs in
figures~\ref{sfig:it2} and~\ref{sfig:it2bb}, in which cases the lower bounds are 2 and 3, respectively.

\begin{figure}[htbp]
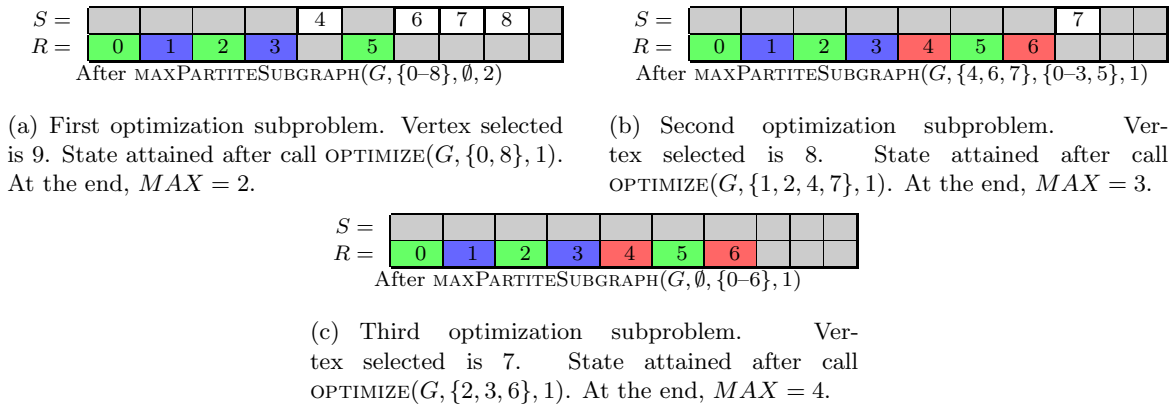

\centering {\scriptsize
		\subfigure[First optimization subproblem. Vertex selected is 9. State attained
		after call {\sc optimize}$(G, \{ 0, 8 \}, 1)$. At the end, $MAX = 2$.]{\label{sfig:it1bb}
		\begin{tabu}{r|[0.7pt]c|[0.7pt]c|[0.7pt]c|[0.7pt]c|[0.7pt]c|[0.7pt]c|[0.7pt]c|[0.7pt]c|[0.7pt]c|[0.7pt]c|[0.7pt]}
		\tabucline[0.7pt]{2-11} 
		$S =$ & \cellcolor[gray]{.8} & \cellcolor[gray]{.8} & \cellcolor[gray]{.8} &
		\cellcolor[gray]{.8} & 4 & \cellcolor[gray]{.8} &
		6 & 7 & 8 &
		\cellcolor[gray]{.8} \\ \tabucline[0.7pt]{2-11} 
		$R =$ & \cellcolor{green!60} 0 & \cellcolor{blue!60} 1 & \cellcolor{green!60} 2 & 
		\cellcolor{blue!60} 3 & \cellcolor[gray]{.8} & \cellcolor{green!60} 5 &
		\cellcolor[gray]{.8} & \cellcolor[gray]{.8} & \cellcolor[gray]{.8} &
		\cellcolor[gray]{.8} \\ \tabucline[0.7pt]{2-11} 
		\multicolumn{11}{c}{After \Call{maxPartiteSubgraph}{$G, \{ 0\text{--}8 \}, \emptyset, 2$}} \\
		\multicolumn{11}{c}{}  \\ 
		\end{tabu}}\quad\quad
		\subfigure[Second optimization subproblem. Vertex selected is 8. State attained after
		call {\sc optimize}$(G, \{ 1, 2, 4, 7 \}, 1)$. At the end, $MAX = 3$.]{\label{sfig:it2bb}
		\begin{tabu}{r|[0.7pt]c|[0.7pt]c|[0.7pt]c|[0.7pt]c|[0.7pt]c|[0.7pt]c|[0.7pt]c|[0.7pt]c|[0.7pt]c|[0.7pt]c|[0.7pt]}
		\tabucline[0.7pt]{2-11} 
		$S =$ & \cellcolor[gray]{.8} & \cellcolor[gray]{.8} & \cellcolor[gray]{.8} &
		\cellcolor[gray]{.8} & \cellcolor[gray]{.8} & \cellcolor[gray]{.8} &
		\cellcolor[gray]{.8} & 7 & \cellcolor[gray]{.8} &
		\cellcolor[gray]{.8} \\ \tabucline[0.7pt]{2-11} 
		$R =$ & \cellcolor{green!60} 0 & \cellcolor{blue!60} 1 & \cellcolor{green!60} 2 & 
		\cellcolor{blue!60} 3 & \cellcolor{red!60} 4 & \cellcolor{green!60} 5 &
		\cellcolor{red!60} 6 & \cellcolor[gray]{.8} & \cellcolor[gray]{.8} &
		\cellcolor[gray]{.8} \\ \tabucline[0.7pt]{2-11} 
		\multicolumn{11}{c}{After \Call{maxPartiteSubgraph}{$G, \{ 4, 6, 7 \}, \{ 0\text{--}3, 5 \}, 1$}} \\
		\multicolumn{11}{c}{}  \\ 
		\end{tabu}}\quad\quad
		\subfigure[Third optimization subproblem. Vertex selected is 7. State attained after
		call {\sc optimize}$(G, \{ 2, 3, 6 \}, 1)$. At the end, $MAX = 4$.]{\label{sfig:it3bb}
		\begin{tabu}{r|[0.7pt]c|[0.7pt]c|[0.7pt]c|[0.7pt]c|[0.7pt]c|[0.7pt]c|[0.7pt]c|[0.7pt]c|[0.7pt]c|[0.7pt]c|[0.7pt]}
		\tabucline[0.7pt]{2-11} 
		$S =$ & \cellcolor[gray]{.8} & \cellcolor[gray]{.8} & \cellcolor[gray]{.8} &
		\cellcolor[gray]{.8} & \cellcolor[gray]{.8} & \cellcolor[gray]{.8} &
		\cellcolor[gray]{.8} & \cellcolor[gray]{.8} & \cellcolor[gray]{.8} &
		\cellcolor[gray]{.8} \\ \tabucline[0.7pt]{2-11} 
		$R =$ & \cellcolor{green!60} 0 & \cellcolor{blue!60} 1 & \cellcolor{green!60} 2 & 
		\cellcolor{blue!60} 3 & \cellcolor{red!60} 4 & \cellcolor{green!60} 5 &
		\cellcolor{red!60} 6 & \cellcolor[gray]{.8} & \cellcolor[gray]{.8} &
		\cellcolor[gray]{.8} \\ \tabucline[0.7pt]{2-11} 
		\multicolumn{11}{c}{After \Call{maxPartiteSubgraph}{$G, \emptyset, \{ 0\text{--}6 \}, 1$}} \\
		\multicolumn{11}{c}{}  \\ 
		\end{tabu}}
}
\caption{Execution of Alg.~\ref{alg:genmethodbb} corresponding to the call
{\sc optimize}$(G, V, 0)$, where $G$ is the graph of Fig.~\ref{fig:graph}.}
\label{fig:execbb}
\end{figure}

\begin{table}[htbp]
\begin{center}
  \begin{tabular}{|c|c|c|} \hline
  Decision subprolem solved & New solution found &
  $MAX$ incremented by \\ \hline \hline 
  $\{ 0, 8 \} = \{ 0, 1, 2, 3, 4, 5, 6, 7, 8 \} \cap N(9)$ & $\{ 0, 9 \}$ & 2 \\
  $\{ 1, 2, 4, 7 \} = \{ 0, 1, 2, 3, 4, 5, 6, 7 \} \cap N(8)$ & $\{ 1, 2, 8 \}$
  & 1 \\
  $\{ 2, 3, 6 \} = \{ 0, 1, 2, 3, 4, 5, 6 \} \cap N(7)$ & $\{ 2, 3, 6, 7 \}$ & 1 \\
  \hline \multicolumn{3}{|l|}{Colors constructed by calls to \Call{maxPartiteSubgraph}{}: $\{ 0, 2, 5
  \}$, $\{ 1, 3 \}$, and $\{ 4, 6 \}$} \\ 
  \multicolumn{3}{|l|}{Uncolored vertices: $6, 7, 8$} \\ \hline
  \end{tabular}
\end{center}
\caption{Summary of execution at Fig.~\ref{fig:execbb}.}
\label{tab:bb}
\end{table}

\section{New Features}
\label{sec:newfeat}

We give in this section more details on the originalities of
Alg.~\ref{alg:genmethod} with respect to the original RD algorithm
proposed in~\cite{Oste02}. As mentioned in the Introduction, some of these new
features appear in B\&B algorithms in the literature. In subsections~\ref{sec:bitpar}
and~\ref{sec:ubvcolor}, we give details on how we adapted them to the RD
method.

\subsection{Bit-Parallelism}
\label{sec:bitpar}

Exploiting bit-level parallelism in sets encoded as bitmaps is central in our
algorithm due to its ability to speedup some operations that are executed very
often during the search. Its effectiveness has already been proved in B\&B
algorithms~\cite{Segu11,Segu14}. In this subsection, we describe the bitmap data structure and the notation adopted for its elementary operations. The application of such operations in our algorithm is left to next subsections.

A {\em bitmap} $B_n$ is a special encoding of a directly addressed set
$B \subseteq [n]$ whose elements are represented as bits in an array. In such
an encoding, for every $i \in [n]$, bit indexed $i$ in $B_n$ is 1 if and
only if $i$ is an element of $B$. For example, the subset $\{ 1, 3, 6\}$ 
of $[8]$ is encoded as $01010010$. Naturally, instead of being viewed as an
array of bits, a bitmap is stored as an array of {\em bitmap nodes} (or simply
{\em nodes}), each of the same size (in bits), denoted by $w$. If we take $w =
4$ in the previous example, the bitmap consists of an array of two nodes: node
of index 0 in the array has value $0101$, and the one of index 1 is
$0010$. Typically, the size $w$ of a node corresponds to the number of bits of a CPU
register. We assume that $w$ is a power of 2 (which is a reasonable assumption
since it equals 32, 64, 128, 256, or 512 in nowadays computers). The size, in
nodes, of $B_n$ is $\lceil n/w \rceil$ and accessing an element $i$ of the set stored in $B_n$ 
implies finding first the correct node and then addressing the exact bit in 
this node. Formulas of bit displacement allow this: for instance 
$i \gg log_2 w$ gives the node index corresponding to the $i$'s bit in the
bitmap.

Typical applications of bitmaps occur in problems
involving the manipulation of sets of vertices or edges of a graph. In our
algorithms, bitmaps are used to store the lines of the adjacency matrix of $G$. For the sake of notation, the bitmap containing the neighborhood of a
vertex $v$ (which is composed by $\lceil n/w \rceil$ nodes) is written
$\Call{neig}{G, v}$. In addition, bitmaps are used to store working sets, as
detailed in next subsections.

Some of the operations performed on bitmaps involve only one element of the set
at hand and cannot benefit from bit-parallelism. Some of them
arise in our algorithms, namely:
\begin{itemize}
\item $\Call{add}{B_n,v}$: adds the element $v < n$ to bitmap $B_n$ (determine
the node containing $v$ and set in this node the corresponding bit to 1).
\item $\Call{rem}{B_n,v}$: removes the element $v < n$ from bitmap $B_n$
(determine the node containing $v$ and set in this node the corresponding bit to 0).
\item $\Call{fsb}{B_n,i,n}$: returns the smallest element in $\{ i, \ldots, n-1
\} \cap B$ (which corresponds to the index of the least significant bit greater
than or equal to $i$ that equals 1 in bitmap $B_n$). Such an operation
is more time consuming with respect to the previous ones since it could incur a search in several nodes of the bitmap. The search in a node $e$ is done by means of the special function \Call{lsb}{$e$} returning the least significant bit of node $e$ (a negative number is returned if
the value of $e$ is zero). In many nowadays processors, such a function is 
provided by the assembler set of instructions. For details on efficient 
implementations of \Call{lsb}{$e$}, we refer the reader to~\cite{Segu11} and
references therein.
\end{itemize}
Bit-parallelism is particularly effective in classical set operations that occur often in our algorithm. The notation for these
operations is the following:
\begin{itemize}
\item $\Call{inter}{B_n,B'_n,n}$: the intersection of the two bitmap encoded
sets $B$ and $B'$ is computed (by making a logical $\&$ on each pair of
corresponding nodes of $B_n$ and $B'_n$) and returned.
\item $\Call{diff}{B_n,B'_n,n}$: the difference $B \setminus B'$ of two
bitmap encoded sets is computed (by making the call $\Call{inter}{B_n,\bar
B'_n,n}$) and returned.
\end{itemize}

\subsection{Upper Bounds by Vertex Coloring}
\label{sec:ubvcolor}

We have studied two different partial coloring heuristics to implement
the generic function \Call{maxPartiteSubgraph}{} in our algorithm. Their
descriptions and bit-parallel implementations are given below. Both are direct
adaptations of greedy heuristics used in B\&B algorithms
like~\cite{Segu11,Tomi07,Tomi13}.

\subsubsection{Greedy Coloring}

This first heuristic is the adaptation of the very classical and simple 
greedy heuristic where the vertices are considered sequentially, each vertex being colored with the 
smallest possible color. However, two characteristics of our implementation
depicted in Alg.~\ref{alg:coloring} have to be mentioned. First, the coloring is
done by colors, as in~\cite{Segu13}, and not by vertices as in classical
implementations of this heuristic (\cite{Tomi07} for instance). The final coloring is the same but this approach is better suited to bit-parallelism.
Second, bit-parallelism is exploited in the set operations 
at lines~\ref{lin:CgetsR1} (copy of a set) and~\ref{lin:Cinters} (set 
difference). Also, vertices still candidates for the current color $d'$ 
are enumerated in an increasing order of vertex indices by the use of function
\Call{fsb}{} (called at lines~\ref{lin:glsb1} and~\ref{lin:glsb2}).

\begin{algorithm}[h]
\caption{Bit-parallel greedy coloring heuristic}
\label{alg:coloring}
\begin{algorithmic}[1]
{\small
\Function{greedyColoring}{$G$, $S$, $R$, $d$}
\State $d' \gets 0$
\While {$d' < d$ and $R \ne \emptyset$}
	\State $S \gets$ a copy of $R$ \label{lin:CgetsR1}
	\State $v \gets \Call{fsb}{S,0,n}$ \label{lin:glsb1}
	\While {$v \geq 0$} \label{lin:Renum} 
		\State $S \gets$ \Call{diff}{$S$, \sc{neig}$(G, v)$, $n$}
		\label{lin:Cinters}
		\State $\Call{rem}{S,v}$ 
        \State $\Call{rem}{R,v}$
        \State $\Call{add}{C,v}$ \label{lin:col11}
		\State $v \gets \Call{fsb}{S,v+1,n}$ \label{lin:glsb2}
	\EndWhile
	\State $d' \gets d' + 1$
\EndWhile
\EndFunction
}
\end{algorithmic}
\end{algorithm}

\subsubsection{Recoloring}

\Call{mcsColoring}{} is a method to improve a greedy coloring successfully 
employed in~\cite{Tomi13}. The idea is to try to assign a lower color to nodes
whose initial color is greater than a given value (based on the value of the 
maximum clique found so far).
In our case, this can be done by applying first the function 
\Call{greedyColoring}{} and then trying to give one of the $d$ first colors to the
remaining nodes, belonging to $R$ at the end of Alg.~\ref{alg:coloring}.
To do so, iteratively for each $v \in S$ we search for a color 
$i \in [d]$ such that $v$ has only one neighbour -- say $u$ -- in the 
corresponding color class (it has at least one neighbour, otherwise $v$
would have been colored with color $i$). If such a color does not exist, we skip the current vertex $v$ and go to 
the next one in $S$. Otherwise, $N(v) \cap R[i] = \{ u \}$ and $u$ has at least
one neighbor in every color smaller than $i$. Thus, we search for a color $j$, $i < j < d$, that could accommodate $u$. 
If such a color is found, move $u$ to color $j$ and $v$ to color $i$ and insert
$v$ in $S$. So, the number of vertices colored (in one of the $d$ first colors)
has been increased. We do not detail the algorithm but, again, the operations $\Call{fsb}{}$,
$\Call{add}{}$, $\Call{rem}{}$, and $\Call{inter}{}$ have to be applied allowing
to benefit from bit parallelism.


\subsection{Recursive Russian Dolls Searches}

Another originality of Alg.~\ref{alg:genmethod} is that each decision subproblem is
solved itself with the same principles of the RD method in function
\Call{decide}{} specified in Subsection~\ref{sec:subproblems} and detailed in
Alg.~\ref{alg:subdoll}. In special, at line~\ref{lin:coloring}, a partial
$(\ell-1)$-coloring of $G[S]$ is performed (moving the colored vertices from $S$ to $R$). If $G[S]$ admits an $(\ell-1)$-coloring ({\em i.e.}, $S=\emptyset$
after this call), then it can be concluded that the decision subproblem at hand
is a ``no'' subproblem without any recursive call. Otherwise, if $S\neq
\emptyset$ after line~\ref{lin:coloring}, then $R$ is the set of vertices pruned
by the pruning rule. In this case, the vertices that remain in $S$ are the
candidates used to generate the recursive calls. So, at each iteration of the loop starting at line~\ref{lin:loopR}, a vertex 
$v$ is chosen from the set of remaining vertices to be added to $R$ and a 
recursive call looks for a clique of size $(\ell-1)$ on $G[R \cap N(v)]$.
If such a clique is found, then $v$ is added to this clique and the function 
returns TRUE. If not, the procedure is repeated until there are no 
vertices left in $S$.

\begin{algorithm}[htb]
\caption{Decision subproblem}
\label{alg:subdoll}
\begin{algorithmic}[1]
{\small
\Function{decide}{$G$, $S$, $\ell$, $Cq$}
\If {$S = \emptyset$} \label{lin:maxfound}
\EndIf
\State $R \gets \emptyset$
\State \Call{maxPartiteSubgraph}{$G$, $S$, $R$, $\ell-1$} \label{lin:coloring}
\While {$S \ne \emptyset$} \label{lin:loopR}
   	\State $v \gets \Call{fsb}{S, 0, n}$
	\State $newS \gets$ \Call{inter}{$R$, {\sc neig}$(G, v)$, $n$}
   	\If {\Call{decide}{$G$, $newS$, $\ell-1$, $Cq$}} \label{lin:reccallred}
		\State $\Call{add}{Cq,v}$
		\State \Return TRUE
	\EndIf
	\State $\Call{add}{R,v}$
	\State $\Call{rem}{S,v}$
\EndWhile
\State \Return TRUE if $\ell = 0$, and FALSE otherwise
\EndFunction
}
\end{algorithmic}
\end{algorithm}

\subsection{Fiding Good Solutions Faster}

The effectiveness of Alg.~\ref{alg:genmethod} depends on the number of
executions of line~\ref{lin:reccallred} of Alg.~\ref{alg:subdoll} or,
equivalently, the number of recursive calls to function \Call{decide}{} when
exploring the generated dolls. In our implementation, we adopt the following
strategy to try to reduce this number of recursive calls. Let us consider that a doll, say $G_i$, is generated
at line~\ref{lin:calldecide} of Alg.~\ref{alg:genmethod}. Also, let a
recursive call in the search associated with $G_i$ be characterized by a pair
$(Cq, newS)$ of a clique $Cq$ of $G_i$ and a set $newS \subseteq V_i$ of
candidates. Since $\omega(G_i) \leq MAX + 1$, the search in $G_i$ can be
interrupted at any point if a better clique of $G$, {\em i.e.} a clique of size at least $MAX + 1$, exist.
For the purpose of possibly interrupting the search in $G_i$, we perform calls
to $\Call{extendClique}{G, newS \cup (V \setminus V_i), Cq}$ to some selected
pairs $(Cq, newS)$. Such a call occurs just before line~\ref{lin:reccallred} of
Alg.~\ref{alg:subdoll} and determines a clique $Cq'$ which contains $Cq$ and
vertices from $newS$ or $V \setminus V_i$. If $|Cq'| > MAX$, then the search in
$G_i$ is interrupted and $MAX$ gets $|Cq'|$. Otherwise, the search in $G_i$
continues. In order to avoid the overhead of an excessive number of calls to
$\Call{extendClique}{}$, they are only performed when $|Cq| \geq MAX/2$.

\section{Experimental Results}
\label{sec:exp}

In this section we provide and analyze results of extensive computational
experiments. Our main goal is to assess whether the Russian Dolls
method is effective at exploiting bit-parallelism and partial colorings to
accelerate the search for a maximum clique when compared to Branch and Bound.
To this end, the computational experiments were carried out with
implementations in the C programming language of
Alg.~\ref{alg:genmethod} and of
Alg.~\ref{alg:genmethodbb}. They are called \Call{RDMC}{} and \Call{PBBMC}{},
respectively. The same routines describer in subsections~\ref{sec:bitpar}
and~\ref{sec:ubvcolor}, for bit-parallel operations and
for partial colorings, respectively, are used in both implementations. This
methodology aims at avoiding the influence of programming settings (such as the programming language adopted or strategies for memory management) in the analysis of the algorithms.
Moreover, we also compare our results with calibration times of previous
experiments in the literature. Even being imprecise in nature, this allows
checking if the computation times of our
implementations are compatible with previous works.

Four versions of each implementation have been tested to compare the
effectiveness of different partial coloring heuristics and vertex orderings. For
each implementation, the notation $X/Y$ stands for the combination coloring
heuristic ($X$) and vertex ordering ($Y$). The partial coloring heuristics
tested were \Call{greedyColoring}{} ($X=1$) and \Call{MCScoloring}{} ($X=2$). At
the beginning of each implementation, vertices of $G$ are renumbered according to a specific order. This renumbering fixes the order employed in 
all executions of partial coloring heuristics and determines the selection
strategies. The orders tested were the one given by the \Call{MCR}{} initial
vertex sorting in~\cite{Tomi07} ($Y=3$) and the one by nonincreasing
degree ($Y=4$). We do not apply an Iterated Local Search as in \cite{Bats13}
since we are interested in studying the behavior of the search procedure,
including its ability to find good solutions fast.

The results were obtained with experiments on a computer running a
64-bit Linux operating system and {\tt gcc} as the C compiler (with compiling
options {\tt -m64 -O3 -msse4.2}). All bit-parallel operations are
implemented with the Streaming SIMD Extensions Instructions for 128 bits. We ran all of our
implementations with DIMACS challenge~\cite{Tric96} and BHOSHLIB (in particular, subset {\tt frb30-15})~\cite{Xu} benchmark graphs and with randomly generated graphs. For
the sake of comparison, experimental results presented in~\cite{Tomi13}
(algorithms \Call{MCR}{} and \Call{MCS}{}), and~\cite{Segu14} (algorithms
\Call{DEF}{} and \Call{RECOL\_N}{}), were adjusted according to the
usual DIMACS challenge methodology~\cite{Tric96} (adopted, for instance,
in~\cite{Tomi13}). The average values ($T_1/T_2$ and $T_1/T_3$)
of the user times in ~\cite{Tomi13} ($T_2$) and~\cite{Segu14} ($T_3$) with our
user times ($T_1$) are shown in Table~\ref{tab:machinebench}. 

\begin{table}[htb]
\begin{center}
\begin{tabular}{|c|c|c|c|c|c|} \hline
Instance & $T_1$ & $T_2$ \cite{Tomi13} & 
$T_1/T_2$ & $T_3$ \cite{Segu14} & $T_1/T_3$ \\
\hline $r$100.5 & -     & 0.00 & -  & 0.00  & -    \\
$r$200.5 & 0.04 & 0.042 & 0.95 & 0.003 & 13.33 \\
$r$300.5 & 0.39 & 0.359 & 1.09 & 0.203 & 1.92 \\
$r$400.5 & 2.41 & 2.21  & 1.09 & 1.186 & 2.03 \\
$r$500.5 & 9.18 & 8.47  & 1.08 & 4.587 & 2.00 \\ \hline
\end{tabular}
\end{center}
\caption{User times used to compute the average factors of $T_1/T_2 = 1.09$ and
$T_1/T_3 = 1.98$ with the benchmark program {\tt dfmax} and the machine benchmark graphs $r$300.5, $r$400.5 e $r$500.5.}
\label{tab:machinebench}
\end{table}


The selection of benchmark graphs is shown in Table~\ref{tab:dimacsi}.
This selection was made avoiding instances with running times too small
 or too large. In Table~\ref{tab:randomi}, the selection of random graphs $n\_p$
 is shown, where $n$ is the number of vertices and $p/100$ is the probability
 that each pair of vertices is picked to define an edge.
These graphs were generated having between 200 and 15000 vertices, and
probabilities from 0.1 up to 0.998. For each configuration, five graphs were
generated.

\begin{table}[htb]
\centering {\footnotesize
\begin{tabular}{c|ccc}
\hline
Instance	& $n$	& dens.	& $\omega(G)$	\\
\hline 
{\tt brock200\_1} & 200 & 0.750 & 21 \\ 
{\tt brock400\_1} & 400 & 0.750 & 27 \\ 
{\tt brock400\_2} & 400 & 0.750 & 29 \\ 
{\tt brock400\_3} & 400 & 0.750 & 31 \\ 
{\tt brock400\_4} & 400 & 0.750 & 33 \\ 
{\tt brock800\_1} & 800 & 0.650 & 23 \\ 
{\tt brock800\_2} & 800 & 0.650 & 24 \\ 
{\tt brock800\_3} & 800 & 0.650 & 25 \\ 
{\tt brock800\_4} & 800 & 0.650 & 26 \\ 
{\tt C250.9} & 250 & 0.900 & 44 \\ 
{\tt DSJC1000.5} & 1000 & 0.500 & 15 \\ 
{\tt DSJC500.5} & 500 & 0.500 & 13 \\ 
{\tt frb30-15-1} & 450 & 0.820 & 30 \\ 
{\tt frb30-15-2} & 450 & 0.820 & 30 \\ 
{\tt frb30-15-3} & 450 & 0.820 & 30 \\
{\tt frb30-15-4} & 450 & 0.820 & 30 \\ 
\hline 
\end{tabular}\quad\quad
\begin{tabular}{c|ccc}
\hline
Instance	& $n$	& dens.	& $\omega(G)$	\\
\hline 
{\tt frb30-15-5} & 450 & 0.820 & 30 \\ 
{\tt gen200\_p0.9\_44} & 200 & 0.900 & 44 \\ 
{\tt gen400\_p0.9\_55} & 400 & 0.900 & 55 \\ 
{\tt gen400\_p0.9\_65} & 400 & 0.900 & 65 \\ 
{\tt MANN\_a27} & 378 & 0.990 & 126 \\ 
{\tt MANN\_a45} & 1035 & 1.000 & 345 \\ 
{\tt p\_hat1000-1} & 1000 & 0.245 & 10 \\ 
{\tt p\_hat1000-2} & 1000 & 0.490 & 46 \\ 
{\tt p\_hat1500-1} & 1500 & 0.253 & 12 \\ 
{\tt p\_hat300-3} & 300 & 0.740 & 36 \\ 
{\tt p\_hat700-2} & 700 & 0.500 & 44 \\ 
{\tt p\_hat700-3} & 700 & 0.750 & 62 \\ 
{\tt san400\_0.9\_1} & 400 & 0.900 & 100 \\ 
{\tt sanr200\_0.9} & 200 & 0.890 & 42 \\ 
{\tt sanr400\_0.5} & 400 & 0.500 & 13 \\ 
{\tt sanr400\_0.7} & 400 & 0.700 & 21 \\
\hline 
\end{tabular}}
\caption{Selected benchmark graphs and their numbers of vertices,
densities, and clique numbers.}
\label{tab:dimacsi}
\end{table}

\begin{table}[htb]
\centering {\footnotesize
\begin{tabular}{c|ccc}
\hline
Instance	& $n$	& dens.	& $\omega(G)$	\\
\hline
{\tt 200\_70} & 200 & 0.700 & 18 \\ 
{\tt 200\_80} & 200 & 0.800 & 25--26 \\ 
{\tt 200\_90} & 200 & 0.900 & 40 \\ 
{\tt 200\_95} & 200 & 0.950 & 60--62 \\ 
{\tt 300\_65} & 300 & 0.650 & 17 \\ 
{\tt 300\_70} & 300 & 0.700 & 20 \\ 
{\tt 300\_80} & 300 & 0.800 & 28--29 \\ 
{\tt 300\_98} & 300 & 0.980 & 116--121 \\ 
{\tt 500\_50} & 500 & 0.500 & 13 \\ 
{\tt 500\_60} & 500 & 0.600 & 17 \\ 
{\tt 500\_65} & 500 & 0.650 & 19--20 \\ 
\hline 
\end{tabular}\quad\quad
\begin{tabular}{c|ccc}
\hline
Instance	& $n$	& dens.	& $\omega(G)$	\\
\hline
{\tt 500\_70} & 500 & 0.700 & 22--23 \\ 
{\tt 500\_994} & 500 & 0.994 & 263--270 \\ 
{\tt 1000\_30} & 1000 & 0.300 & 9 \\ 
{\tt 1000\_40} & 1000 & 0.400 & 12 \\ 
{\tt 1000\_50} & 1000 & 0.500 & 15 \\ 
{\tt 1000\_998} & 1000 & 0.998 & 606--613 \\ 
{\tt 5000\_10} & 5000 & 0.100 & 7 \\ 
{\tt 5000\_20} & 5000 & 0.200 & 9 \\ 
{\tt 5000\_30} & 5000 & 0.300 & 12 \\ 
{\tt 10000\_10} & 10000 & 0.100 & 7--8 \\ 
{\tt 15000\_10} & 15000 & 0.100 & 8 \\ 
\hline 
\end{tabular}}
\caption{Selected random graphs $n\_p$ and their clique numbers. For every
configuration, the smallest and greatest clique numbers among the several instances considered are given.}
\label{tab:randomi}
\end{table}

Tables~\ref{tab:dimacs} (for benchmark graphs) and~\ref{tab:random} (for
random graphs) show two measures related to the total number of function calls
performed to determine the initial coloring of decision
(line~\ref{lin:coloring} of Alg.~\ref{alg:subdoll}) or optimization
(line~\ref{lin:maxPartbb1} of Alg.~\ref{alg:genmethodbb}) subproblems.
The first one is its total number (referred to as ``all'') and the second, the
number of calls that result in nonempty sets of uncolored vertices (or, equivalently, the number of
nonleaf nodes of the corresponding search tree, referred to as ``ne'').
Also the CPU user times (in seconds) measured in the experiments are shown in
these tables. The rows are sorted in a nondecreasing order of graph density. For
each graph, the results with the four possible $X/Y$ configurations (lines 1/3, 1/4,
2/3, 2/4) of each implementation (columns \Call{PBBMC}{} or \Call{RDMC}{}) are
given. Computation times available in~\cite{Tomi13} and~\cite{Segu14} are also
presented. In particular, computation times extracted from~\cite{Segu14} are the
ones corresponding to the versions \Call{DEF}{} and \Call{RECOL\_N}{}.
The times spent in the renumbering
procedure are not included in the reported computation times for our
implementations. The symbol in the last column indicates the relative
performance $s = $ (\Call{PBBMC}{} computation time/\Call{RDMC}{} computation
time), as follows: ``--'' for $s \in (0.95,1.05)$; ``$\star$'',
``$\star$$\star$'', and ``$\star$$\star$$\star$'' for the cases when $s$ is in the intervals
$[1.05,1.5)$, $[1.5,2)$, and $[2,\infty)$, respectively; and ``$\circ$'',
``$\circ$$\circ$'', and ``$\circ$$\circ$$\circ$'' for the cases when $s$ belongs
to $(0.67,0.95]$, $(0.5,0.67]$, and $(0,0.5]$, respectively. 

\begin{table}
\begin{center}
\begin{tabular}{|cc|} \hline
{\bf Instance} & {\bf Time} \\ \hline \hline
{\tt 5000\_10}  & 3		\\
{\tt 10000\_10} & 27.39	\\
{\tt 15000\_10} & 92.60	\\
{\tt 5000\_20}  & 9.80	\\
{\tt 5000\_30}  & 20.39	\\
{\tt 1000\_998} & 1.20	\\ \hline
\end{tabular}
\end{center}
\caption{Cases in which the renumbering time for the MCR order is significant
with respect to computation time.}
\label{tab:rtime}
\end{table}

In the results with our implementations, we observe that the order in which the
vertices are considered has a great impact on the computation time in several cases. There
are two parameters that influences this fact. First, with respect to the
renumbering procedure, it has small influence for graphs with density below
75\%. However, the ordering of vertices by noincreasing degree tends to
perform better than the MCR one with a few exceptions for graphs of density at
least 70\%. The more remarkable examples are {\tt frb30-15-3}, {\tt
frb30-15-4}, {\tt frb30-15-5}, {\tt 300\_98}, {\tt 500\_994}, and {\tt 1000\_998}. 
In addition, the time spent in the renumbering procedure is negligible for the
noincreasing degree order, but this is not the case for the MCR in the cases shown in Table~\ref{tab:rtime}.
Second, with respect to the coloring heuristic adopted, the version with
recoloring is effective in determining better partial coloring and, thus, in reducing the
number of explored nodes. This behavior is expected and has already been
observed in~\cite{Tomi13} and~\cite{Segu14}. However, even if the recoloring
procedure is applied only for the nodes in the $MAX/4$ highest levels of the
search tree (which is the case in our experiments), its computation time is too
high if the graph is not sufficiently dense. It is worth remarking that the benefits of the bit-parallelism are severely reduced when the recoloring procedure is used. 

A time comparison with leading previous works indicates that our implementation
of \Call{PBBMC}{} is competitive. In almost all cases, the computation times
of our implementation are much smaller than the calibrated times from the
literature. Even if this comparison suffers from factors due to the programming
environment like programming language, compiler, register length, operating
system and so on, there are convincing evidences that our implementation is
very efficient.

The comparison between \Call{PBBMC}{} and \Call{RDMC}{} is the main objective
of the experiments. Since both B\&B and RD approaches employ a depth-first
strategy, there is a number of explored nodes of the corresponding search trees
whose upper bound obtainable with the coloring heuristics are smaller than the
optimum value $\omega(G)$. This may occur only while the optimum value is not
found. For this reason, the smaller is the number of explored nodes until the
optimum value is found the better is the relative performance of the algorithm.
An evidence of this phenomenon is that there exists a one-to-one
correspondence between leading implementations (\Call{PBBMC}{} or \Call{RDMC}{})
and smaller number of explored subproblems, with the only exception of the
version 1/3 for {\tt san400\_0.9\_1}. We can observe that {\tt C250.9}, {\tt p\_hat300-3}, {\tt p\_hat700-3}, and {\tt p\_hat1000-2}
are instances of DIMACS benchmark graphs whose running times of \Call{RDMC}{}
are significantly faster, which shows its effectiveness when the density is at
least 80\%. Moreover, there are the special cases {\tt gen400\_0.9\_55} and {\tt
gen400\_0.9\_65} for which only \Call{RDMC}{} is able to finish execution
within the time limit of 18000 seconds. This conclusion is corroborated by the
results with random graphs, in which cases the best ratios of improvement are achieved.
In spite of this, it should be noticed that
{\tt brock800\_X} are cases with large running times
for which \Call{PBBMC}{} is faster than
\Call{RDMC}{}.

{\footnotesize
\begin{longtable}{|c|rrrr||rrrl|}
\caption[]{Comparison of the number of generated subproblems and running times
of our implementations on benchmark graphs.
Execution times from~\cite{Tomi13,Segu14} are adjusted according to the 
respective factors listed in Table~\ref{tab:machinebench}. The computation times
for~\cite{Segu14} correspond to the version \Call{DEF}{} and \Call{RECOL\_N}{}.
A blank entry means ``information not available.''\label{tab:dimacs}}\\
\hline \hline
\multirow{2}{*}{\bf Instance} & \multicolumn{4}{|c||}{\bf N. of subps $\times
10^{-3}$} & \multicolumn{4}{|c|}{\bf Computation Time (sec.)}  \\
&\multicolumn{2}{c}{\Call{PBBMC}{}} & \multicolumn{2}{c||}{\Call{RDMC}{}} & &
\Call{PBBMC}{} & \Call{RDMC}{} & \\ \hline 
\multirow{4}{*}{\em Graph} & 1/3-all & 1/3-ne &	1/3-all &	1/3-ne &
\multirow{2}{*}{\cite{Tomi13}}     	& 1/3 & 1/3 & \\
& 1/4-all & 1/4-ne 	&	1/4-all &	1/4-ne  & 	& 1/4 &	1/4 &	   	\\
\cline{2-9} 
& 2/3-all & 2/3-ne	&	2/3-all &	2/3-ne & \multirow{2}{*}{\cite{Segu14}} & 2/3 & 2/3 &  	\\
& 2/4-all & 2/4-ne 	&	2/4-all &	2/4-ne &  & 2/4 &	2/4 &	  	\\
\hline \hline
\endfirsthead
\caption[]{(continued)}\\
\hline \hline
\multirow{2}{*}{\bf Instance} & \multicolumn{4}{|c||}{\bf N. of subps $\times
10^{-3}$} & \multicolumn{4}{|c|}{\bf Computation Time (sec.)}  \\
&\multicolumn{2}{c}{\Call{PBBMC}{}} & \multicolumn{2}{c||}{\Call{RDMC}{}} & &
\Call{PBBMC}{} & \Call{RDMC}{} & \\ \hline 
\multirow{4}{*}{\em Graph} & 1/3-all & 1/3-ne &	1/3-all &	1/3-ne &
\multirow{2}{*}{\cite{Tomi13}}     	& 1/3 & 1/3 & \\
& 1/4-all & 1/4-ne 	&	1/4-all &	1/4-ne  & 	& 1/4 &	1/4 &	   	\\
\cline{2-9} 
& 2/3-all & 2/3-ne	&	2/3-all &	2/3-ne & \multirow{2}{*}{\cite{Segu14}} & 2/3 & 2/3 &  	\\
& 2/4-all & 2/4-ne 	&	2/4-all &	2/4-ne &  & 2/4 &	2/4 &	  	\\
\hline \hline
\endhead
\rowcolor[gray]{.8}  & 169.0 & 38.23 & 164.3 & 36.59 & -- & 0.141 & 0.140 & -- \\* 
\rowcolor[gray]{.8}  & 142.6 & 30.01 & 139.0 & 28.70 & -- & 0.165 & 0.164 & -- \\*
\cline{2-9}
\rowcolor[gray]{.8}  & 176.7 & 37.97 & 174.1 & 37.20 & 0.506 & 0.146 & 0.147 & -- \\* 
\rowcolor[gray]{.8} \multirow{-4}{*}{\tt p\_hat1000-1}  & 151.9 & 30.31 & 149.7 & 29.58 & 0.584 & 0.171 & 0.172 & -- \\
\hline

 & 1056 & 147.8 & 1451 & 229.1 & 5.55 & 1.30 & 1.65 & $\circ$ \\* 
 & 951.1 & 125.7 & 1377 & 220.1 & 4.25 & 1.51 & 1.87 & $\circ$ \\*
\cline{2-9}
 & 1487 & 220.6 & 1128 & 175.6 & 5.08 & 1.70 & 1.36 & $\star$ \\* 
\multirow{-4}{*}{\tt p\_hat1500-1}  & 1398 & 205.9 & 1033 & 154.6 & 5.44 & 1.93 & 1.63 & $\star$ \\
\hline

\rowcolor[gray]{.8}  & 22973 & 6017 & 13559 & 3451 & 2653 & 77.43 & 48.64 & $\star$$\star$ \\* 
\rowcolor[gray]{.8}  & 14598 & 3305 & 8690 & 1899  & 240.8 & 80.93 & 50.65 & $\star$$\star$ \\*
\cline{2-9}
\rowcolor[gray]{.8}  & 39574 & 11504 & 20594 & 5789 & 284.1 & 127.7 & 71.43 & $\star$$\star$ \\* 
\rowcolor[gray]{.8} \multirow{-4}{*}{\tt p\_hat1000-2}  & 24645 & 6402 & 12895 & 3217 & 245.7 & 133.3 & 74.08 & $\star$$\star$ \\
\hline

 & 76070 & 16081 & 74259 & 15532 & 441.4 & 87.69 & 88.53 & -- \\* 
 & 70427 & 15331 & 68479 & 14768 & 319.3 & 96.28 & 96.12 & -- \\*
\cline{2-9}
 & 79916 & 16289 & 82514 & 16969 & 295.4 & 92.39 & 96.76 & -- \\* 
\multirow{-4}{*}{\tt DSJC1000.5}  & 74569 & 15578 & 76987 & 16257 & 333.8 & 100.9 & 104.6 & -- \\
\hline

\rowcolor[gray]{.8}  & 1084 & 275.7 & 1012 & 255.9 & 4.57 & 0.774 & 0.764 & -- \\* 
\rowcolor[gray]{.8}  & 846.3 & 193.8 & 779.7 & 173.3 & 3.37 & 0.951 & 0.915 & -- \\*
\cline{2-9}
\rowcolor[gray]{.8}  & 1183 & 287.7 & 1102 & 266.0 & 2.13 & 0.846 & 0.838 & -- \\* 
\rowcolor[gray]{.8} \multirow{-4}{*}{\tt DSJC500.5}  & 950.5 & 211.8 & 874.3 & 189.3 & 2.51 & 1.03 & 0.998 & -- \\
\hline

 & 231.3 & 57.09 & 231.3 & 57.09 & 48.39 & 0.738 & 0.769 & -- \\* 
 & 162.5 & 33.74 & 162.5 & 33.74 & 6.10 & 0.839 & 0.856 & -- \\*
\cline{2-9}
 & 741.3 & 207.9 & 345.5 & 95.32 & 5.30 & 2.10 & 1.09 & $\star$$\star$ \\* 
\multirow{-4}{*}{\tt p\_hat700-2}  & 504.8 & 123.7 & 238.6 & 57.25 & 5.08 & 2.39 & 1.23 & $\star$$\star$ \\
\hline

\rowcolor[gray]{.8}  & 250.9 & 62.11 & 236.2 & 57.57 & -- & 0.175 & 0.174 & -- \\* 
\rowcolor[gray]{.8}  & 196.9 & 39.25 & 185.3 & 34.02 & -- & 0.216 & 0.207 & -- \\*
\cline{2-9}
\rowcolor[gray]{.8}  & 266.0 & 63.91 & 204.4 & 44.12 & 0.433 & 0.185 & 0.159 & $\star$ \\* 
\rowcolor[gray]{.8} \multirow{-4}{*}{\tt sanr400\_0.5}  & 212.7 & 41.94 & 164.0 & 26.29 & 0.520 & 0.230 & 0.190 & $\star$ \\
\hline

 & 1924228 & 451474 & 3544359 & 904788 & 19390 & 2975 & 5194 & $\circ$$\circ$ \\* 
 & 1327872 & 277754 & 2488522 & 600952 & 10188 & 3073 & 5321 & $\circ$$\circ$ \\*
\cline{2-9}
 & 2149229 & 487116 & 3891822 & 964831 & 9303 & 3367 & 5673 & $\circ$$\circ$ \\* 
\multirow{-4}{*}{\tt brock800\_1}  & 1546067 & 316388 & 2861531 & 682516 & 9497 & 3491 & 5901 & $\circ$$\circ$ \\
\hline

\rowcolor[gray]{.8}  & 1654952 & 387936 & 4459555 & 1155932 & 17492 & 2599 & 6448 & $\circ$$\circ$$\circ$ \\* 
\rowcolor[gray]{.8}  & 1172568 & 251211 & 3189649 & 806701 & 9121 & 2679 & 6636 & $\circ$$\circ$$\circ$ \\*
\cline{2-9}
\rowcolor[gray]{.8}  & 2063020 & 474353 & 4460120 & 1118033 & 8533 & 3200 & 6388 & $\circ$$\circ$ \\* 
\rowcolor[gray]{.8} \multirow{-4}{*}{\tt brock800\_2}  & 1522715 & 326797 & 3325413 & 816754 & 8378 & 3338 & 6700 & $\circ$$\circ$$\circ$ \\
\hline

 & 1065164 & 240368 & 1912833 & 464107 & 11829 & 1719 & 2968 & $\circ$$\circ$ \\* 
 & 757036 & 159313 & 1299528 & 280306  & 6272 & 1774 & 3013 & $\circ$$\circ$ \\*
\cline{2-9}
 & 1228395 & 269984 & 1951966 & 454995 & 5574 & 1973 & 3051 & $\circ$$\circ$ \\* 
\multirow{-4}{*}{\tt brock800\_3}  & 909264 & 187194 & 1395656 & 295756 & 5396 & 2068 & 3135 & $\circ$$\circ$ \\
\hline

\rowcolor[gray]{.8}  & 679270 & 145266 & 2833636 & 734803 & 8217 & 1171 & 4082 & $\circ$$\circ$$\circ$ \\* 
\rowcolor[gray]{.8}  & 480282 & 96096 & 1996967 & 497131  & 4356 & 1200 & 4261 & $\circ$$\circ$$\circ$ \\*
\cline{2-9}
\rowcolor[gray]{.8}  & 2242760 & 536283 & 2641180 & 649939 & 3987 & 3315 & 3886 & $\circ$ \\* 
\rowcolor[gray]{.8} \multirow{-4}{*}{\tt brock800\_4}  & 1695941 & 403075 & 1935412 & 458605 & 3913 & 3475 & 4096 & $\circ$ \\
\hline

 & 55258 & 16573 & 55275 & 16581 & 413.1 & 55.48 & 57.18 & -- \\* 
 & 35267 & 9666 & 35240 & 9656  & 197.2 & 64.46 & 64.91 & -- \\*
\cline{2-9}
 & 64523 & 18819 & 67641 & 19935 & 129.9 & 64.34 & 68.74 & $\circ$ \\* 
\multirow{-4}{*}{\tt sanr400\_0.7}  & 43415 & 11697 & 45552 & 12459 & 139.1 & 75.22 & 79.03 & -- \\
\hline

\rowcolor[gray]{.8}  & 407.5 & 127.4 & 154.2 & 44.19 & 11.77 & 0.607 & 0.260 & $\star$$\star$$\star$ \\* 
\rowcolor[gray]{.8}  & 243.8 & 69.22 & 93.78 & 22.74 & 2.72 & 0.758 & 0.317 & $\star$$\star$$\star$ \\*
\cline{2-9}
\rowcolor[gray]{.8}  & 762.1 & 258.4 & 242.7 & 75.07 & 1.79 & 1.07 & 0.402 & $\star$$\star$$\star$ \\* 
\rowcolor[gray]{.8} \multirow{-4}{*}{\tt p\_hat300-3}  & 458.6 & 142.7 & 149.9 & 40.95 & 1.94 & 1.36 & 0.489 & $\star$$\star$$\star$ \\
\hline

 & 307.2 & 101.7 & 460.0 & 160.5 & 1.87 & 0.214 & 0.312 & $\circ$ \\* 
 & 181.7 & 54.92 & 271.5 & 86.64 & 0.937 & 0.286 & 0.424 & $\circ$ \\*
\cline{2-9}
 & 391.0 & 129.2 & 482.5 & 164.6 & 0.495 & 0.269 & 0.328 & $\circ$ \\* 
\multirow{-4}{*}{\tt brock200\_1}  & 254.1 & 77.52 & 306.3 & 95.91 & 0.546 & 0.365 & 0.456 & $\circ$ \\
\hline

\rowcolor[gray]{.8}  & 178635 & 53031 & 218166 & 66623 & 1930 & 216.6 & 265.8 & $\circ$ \\* 
\rowcolor[gray]{.8}  & 104981 & 29024 & 123173 & 34472 & 755.3 & 237.4 & 284.9 & $\circ$ \\*
\cline{2-9}
\rowcolor[gray]{.8}  & 203454 & 58646 & 182531 & 52453 & 495 & 245.5 & 232.4 & $\star$ \\* 
\rowcolor[gray]{.8} \multirow{-4}{*}{\tt brock400\_1}  & 127091 & 34221 & 109789 & 28684 & 501.5 & 276.4 & 251.2 & $\star$ \\
\hline

 & 70985 & 20478 & 163578 & 51435 & 791.3 & 92.61 & 197.2 & $\circ$$\circ$$\circ$ \\* 
 & 45361 & 12912 & 92685 & 27352  & 323.7 & 99.46 & 208.9 & $\circ$$\circ$$\circ$ \\*
\cline{2-9}
 & 122378 & 35356 & 157131 & 45587 & 210.4 & 151.8 & 198.4 & $\circ$ \\* 
\multirow{-4}{*}{\tt brock400\_2}  & 82234 & 23157 & 95762 & 25701 & 210.8 & 166.8 & 214.9 & $\circ$ \\
\hline

\rowcolor[gray]{.8}  & 127946 & 39572 & 119906 & 37436 & 1308 & 146.8 & 143.9 & -- \\* 
\rowcolor[gray]{.8}  & 75133 & 21945 & 66964 & 19342  & 510.1 & 162.5 & 153.8 & $\star$ \\*
\cline{2-9}
\rowcolor[gray]{.8}  & 168342 & 50294 & 113363 & 34720 & 332.0 & 193.2 & 133.6 & $\star$ \\* 
\rowcolor[gray]{.8} \multirow{-4}{*}{\tt brock400\_3}  & 105186 & 29763 & 69544 & 20137 & 342.9 & 216.9 & 146.7 & $\star$ \\
\hline

 & 60763 & 17870 & 65939 & 20557 & 696.5 & 74.73 & 81.08 & $\circ$ \\* 
 & 35907 & 9813 & 36382 & 10389  & 270.3 & 81.78 & 85.50 & -- \\*
\cline{2-9}
 & 68010 & 19515 & 51187 & 15660 & 206.5 & 83.50 & 61.72 & $\star$ \\* 
\multirow{-4}{*}{\tt brock400\_4}  & 43000 & 11533 & 31194 & 8985 & 213.2 & 93.54 & 67.48 & $\star$ \\
\hline

\rowcolor[gray]{.8}  & 166590 & 42760 & 96961 & 24375 & 74323 & 649.4 & 403.3 & $\star$$\star$ \\* 
\rowcolor[gray]{.8}  & 92589 & 21232 & 54109 & 12126  & 2607 & 670.1 & 408.4 & $\star$$\star$ \\*
\cline{2-9}
\rowcolor[gray]{.8}  & 306990 & 84282 & 151510 & 40865 & 2578 & 1173 & 622.0 & $\star$$\star$ \\* 
\rowcolor[gray]{.8} \multirow{-4}{*}{\tt p\_hat700-3}  & 174538 & 43336 & 86027 & 20922 & 2296 & 1231 & 644.7 & $\star$$\star$ \\
\hline

 & 561752 & 195532 & 403867 & 141579 & -- & 970.4 & 727.8 & $\star$ \\* 
 & 297794 & 106338 & 227680 & 81421  & -- & 809.7 & 639.4 & $\star$ \\*
\cline{2-9}
 & 501918 & 173172 & 583655 & 205002 & 2361 & 899.3 & 1057 & $\circ$ \\* 
\multirow{-4}{*}{\tt frb30-15-1}  & 304227 & 106217 & 358024 & 128089 & 3057 & 824.4 & 961.4 & $\circ$ \\
\hline

\rowcolor[gray]{.8}  & 564777 & 196671 & 587228 & 205260 & -- & 967.5 & 1039 & $\circ$ \\* 
\rowcolor[gray]{.8}  & 323762 & 115112 & 337370 & 120645 & -- & 865.0 & 899.1 & -- \\*
\cline{2-9}
\rowcolor[gray]{.8}  & 512867 & 183336 & 672671 & 240659 & 2725 & 897.5 & 1193 & $\circ$ \\* 
\rowcolor[gray]{.8} \multirow{-4}{*}{\tt frb30-15-2}  & 323978 & 120593 & 436865 & 162446 & 2102 & 810.9 & 1052 & $\circ$ \\
\hline

 & 448754 & 156829 & 741805 & 268482 & -- & 779.7 & 1282 & $\circ$$\circ$ \\* 
 & 243166 & 87841 & 415467 & 156231  & -- & 652.5 & 1028 & $\circ$$\circ$ \\*
\cline{2-9}
 & 163409 & 59582 & 319548 & 121538  & 1551 & 294.2 & 557.2 & $\circ$$\circ$ \\* 
\multirow{-4}{*}{\tt frb30-15-3}  & 111497 & 41752 & 216080 & 84317 & 991.3 & 276.3 & 482.2 & $\circ$$\circ$ \\
\hline

\rowcolor[gray]{.8}  & 1781554 & 627424 & 843752 & 287111 & -- & 2961 & 1540 & $\star$$\star$ \\* 
\rowcolor[gray]{.8}  & 998141 & 364563 & 440556 & 153258  & -- & 2548 & 1259 & $\star$$\star$$\star$ \\*
\cline{2-9}
\rowcolor[gray]{.8}  & 1125727 & 399811 & 483553 & 163851 & 4694 & 1905 & 901.6 & $\star$$\star$$\star$ \\* 
\rowcolor[gray]{.8} \multirow{-4}{*}{\tt frb30-15-4}  & 753613 & 277360 & 297691 & 103577 & 3263 & 1746 & 789.2 & $\star$$\star$$\star$ \\
\hline

 & 664074 & 235213 & 995328 & 360589 & -- & 1170 & 1754 & $\circ$$\circ$ \\* 
 & 334927 & 123988 & 522826 & 198063 & -- & 897.5 & 1329 & $\circ$ \\*
\cline{2-9}
 & 342092 & 126432 & 623818 & 236143 & 3412 & 601.5 & 1075 & $\circ$$\circ$ \\* 
\multirow{-4}{*}{\tt frb30-15-5}  & 239928 & 91485 & 457648 & 177749 & 2166 & 546.5 & 963.2 & $\circ$$\circ$ \\
\hline

\rowcolor[gray]{.8}  & 11607 & 4539 & 10304 & 3998 & 315.0 & 14.69 & 13.18 & $\star$ \\* 
\rowcolor[gray]{.8}  & 3465 & 1439 & 2993 & 1233  & 44.69 & 13.67 & 12.28 & $\star$ \\*
\cline{2-9}
\rowcolor[gray]{.8}  & 14729 & 5710 & 14528 & 5683 & 27.16 & 18.22 & 18.34 & -- \\* 
\rowcolor[gray]{.8} \multirow{-4}{*}{\tt sanr200\_0.9}  & 6079 & 2367 & 5874 & 2325 & 21.60 & 20.76 & 20.66 & -- \\
\hline

 & 991805 & 388397 & 781787 & 302033 & 48193 & 1240 & 1183 & -- \\* 
 & 316861 & 131958 & 245797 & 100860 & 3550 & 1271 & 1032 & $\star$ \\*
\cline{2-9}
 & 1137737 & 445366 & 972115 & 378289 & 2361 & 1474 & 1303 & $\star$ \\* 
\multirow{-4}{*}{\tt C250.9}  & 459609 & 182377 & 390771 & 154130 & 1964 & 1666 & 1553 & $\star$ \\
\hline

\rowcolor[gray]{.8}  & 197.1 & 82.04 & 687.2 & 286.7 & 5.87 & 0.259 & 0.976 & $\circ$$\circ$$\circ$ \\* 
\rowcolor[gray]{.8}  & 66.56 & 30.46 & 205.6 & 93.02 & 0.512 & 0.218 & 0.891 & $\circ$$\circ$$\circ$ \\*
\cline{2-9}
\rowcolor[gray]{.8}  & 329.5 & 146.1 & 966.1 & 427.7 & 0.766 & 0.432 & 1.39 & $\circ$$\circ$$\circ$ \\* 
\rowcolor[gray]{.8} \multirow{-4}{*}{\tt gen200\_p0.9\_44}  & 142.2 & 66.60 & 361.4 & 167.3 & 0.493 & 0.437 & 1.22 & $\circ$$\circ$$\circ$ \\
\hline

 & -- & -- & 2364038 & 910227  & 6373176 & $>$ 18000 & 6289 &
 $\star$$\star$$\star$ \\* 
 & -- & -- & 666569 & 261728 & 63689 & $>$ 18000 & 3691 & $\star$$\star$$\star$
 \\* 
 \cline{2-9} & -- & -- & -- & -- & -- & $>$ 18000 & $>$ 18000 & -- \\* 
\multirow{-4}{*}{\tt gen400\_p0.9\_55}  & -- & -- & -- & -- & -- & $>$ 18000 & $>$ 18000 & -- \\
\hline

\rowcolor[gray]{.8}  & -- & -- & 536110 & 193972 & -- & $>$ 18000 & 1527 &
$\star$$\star$$\star$ \\* 
\rowcolor[gray]{.8}  & -- & -- & 116621 & 44894  & 165240 & $>$ 18000 & 749.1 & $\star$$\star$$\star$ \\*
\cline{2-9}
\rowcolor[gray]{.8}  & -- & -- & -- & -- & -- & $>$ 18000 & $>$ 18000 & -- \\* 
\rowcolor[gray]{.8} \multirow{-4}{*}{\tt gen400\_p0.9\_65}  & -- & -- & -- & -- & -- & $>$ 18000 & $>$ 18000 & -- \\
\hline

 & 377.0 & 159.1 & 302.0 & 83.82 & 3.70 & 0.745 & 1.00 & $\circ$ \\* 
 & 52.64 & 26.83 & 28.50 & 9.85  & 0.109 & 0.340 & 0.312 & $\star$ \\*
\cline{2-9}
 & 5275 & 2194 & 30308 & 11794   & 60.44 & 13.36 & 71.00 & $\circ$$\circ$$\circ$ \\* 
\multirow{-4}{*}{\tt san400\_0.9\_1}  & 1957 & 909.4 & 5088 & 2170 & 15.54 & 8.58 & 32.82 & $\circ$$\circ$$\circ$ \\
\hline

\rowcolor[gray]{.8}  & 37.90 & 18.30 & 40.40 & 19.39 & 2.72 & 0.194 & 0.213 & $\circ$ \\* 
\rowcolor[gray]{.8}  & 8.92 & 4.67 & 9.94 & 5.32  & 0.872 & 0.123 & 0.139 & $\circ$ \\*
\cline{2-9}
\rowcolor[gray]{.8}  & 37.90 & 18.30 & 40.40 & 19.39 & 0.453 & 0.178 & 0.198 & $\circ$ \\* 
\rowcolor[gray]{.8} \multirow{-4}{*}{\tt MANN\_a27}  & 8.92 & 4.67 & 9.94 & 5.32 & 0.310 & 0.142 & 0.163 & $\circ$ \\
\hline

 & 2952 & 1081 & 3827 & 1485 & 3368 & 71.05 & 94.25 & $\circ$ \\* 
 & 242.9 & 118.7 & 452.9 & 225.6  & 306.2 & 18.40 & 31.73 & $\circ$$\circ$ \\*
\cline{2-9}
 & 2952 & 1081 & 3827 & 1485 & 220.7 & 63.46 & 84.82 & $\circ$ \\* 
\multirow{-4}{*}{\tt MANN\_a45}  & 242.9 & 118.7 & 452.9 & 225.6 & 53.16 & 18.10 & 39.18 & $\circ$$\circ$$\circ$ \\
\hline
\hline
\end{longtable}}

{\footnotesize
\begin{longtable}{|c|rrrr||rrrl|}
\caption[]{Comparison of the number of generated subproblems and running times
of our implementations on benchmark graphs.
Execution times from~\cite{Tomi13,Segu14} are adjusted according to the 
respective factors listed in Table~\ref{tab:machinebench}. The computation times
for~\cite{Segu14} correspond to the version \Call{DEF}{} and \Call{RECOL\_N}{}.
A blank entry means ``information not available.''\label{tab:random}}\\
\hline \hline
\multirow{2}{*}{\bf Instance} & \multicolumn{4}{|c||}{\bf N. of subps $\times
10^{-3}$} & \multicolumn{4}{|c|}{\bf Computation Time (sec.)}  \\
&\multicolumn{2}{c}{\Call{PBBMC}{}} & \multicolumn{2}{c||}{\Call{RDMC}{}} & &
\Call{PBBMC}{} & \Call{RDMC}{} & \\ \hline 
\multirow{4}{*}{\em Graph} & 1/3-all & 1/3-ne &	1/3-all &	1/3-ne &
\multirow{2}{*}{\cite{Tomi13}}     	& 1/3 & 1/3 & \\
& 1/4-all & 1/4-ne 	&	1/4-all &	1/4-ne  & 	& 1/4 &	1/4 &	   	\\
\cline{2-9} 
& 2/3-all & 2/3-ne	&	2/3-all &	2/3-ne & \multirow{2}{*}{\cite{Segu14}} & 2/3 & 2/3 &  	\\
& 2/4-all & 2/4-ne 	&	2/4-all &	2/4-ne &  & 2/4 &	2/4 &	  	\\
\hline \hline
\endfirsthead
\caption[]{(continued)}\\
\hline \hline
\multirow{2}{*}{\bf Instance} & \multicolumn{4}{|c||}{\bf N. of subps $\times
10^{-3}$} & \multicolumn{4}{|c|}{\bf Computation Time (sec.)}  \\
&\multicolumn{2}{c}{\Call{PBBMC}{}} & \multicolumn{2}{c||}{\Call{RDMC}{}} & &
\Call{PBBMC}{} & \Call{RDMC}{} & \\ \hline 
\multirow{4}{*}{\em Graph} & 1/3-all & 1/3-ne &	1/3-all &	1/3-ne &
\multirow{2}{*}{\cite{Tomi13}}     	& 1/3 & 1/3 & \\
& 1/4-all & 1/4-ne 	&	1/4-all &	1/4-ne  & 	& 1/4 &	1/4 &	   	\\
\cline{2-9} 
& 2/3-all & 2/3-ne	&	2/3-all &	2/3-ne & \multirow{2}{*}{\cite{Segu14}} & 2/3 & 2/3 &  	\\
& 2/4-all & 2/4-ne 	&	2/4-all &	2/4-ne &  & 2/4 &	2/4 &	  	\\
\hline \hline
\endhead
\rowcolor[gray]{.8}  & 522.8 & 9.95 & 542.7 & 8.77 & 5.77 & 1.22 & 1.23 & -- \\* 
\rowcolor[gray]{.8}  & 521.9 & 9.86 & 542.0 & 8.73 & 3.59 & 1.24 & 1.25 & -- \\*
\cline{2-9}
\rowcolor[gray]{.8}  & 532.4 & 10.91 & 555.9 & 10.72 & 2.89 & 1.25 & 1.26 & -- \\* 
\rowcolor[gray]{.8} \multirow{-4}{*}{\tt 5000\_10}  & 531.4 & 10.80 & 555.3 & 10.66 & 3.28 & 1.27 & 1.29 & -- \\
\hline

 & 4707 & 673.4 & 5159 & 832.3 & 109 & 21.84 & 21.92 & -- \\* 
 & 4645 & 629.4 & 5128 & 812.7 & 65.40 & 21.98 & 22.04 & -- \\*
\cline{2-9}
 & 4833 & 688.9 & 5056 & 769.4 & 49.85 & 22.34 & 22.00 & -- \\* 
\multirow{-4}{*}{\tt 10000\_10}  & 4778 & 650.5 & 5014 & 740.1 & 51.59 & 22.50 & 22.15 & -- \\
\hline

\rowcolor[gray]{.8}  & 20829 & 2648 & 20553 & 3106 & 556.9 & 121.0 & 116.2 & -- \\* 
\rowcolor[gray]{.8}  & 18372 & 1825 & 18072 & 2369 & 356.4 & 120.3 & 116.5 & -- \\*
\cline{2-9}
\rowcolor[gray]{.8}  & 21468 & 2631 & 22900 & 3290 & 284.3 & 124.0 & 123.3 & -- \\* 
\rowcolor[gray]{.8} \multirow{-4}{*}{\tt 15000\_10}  & 19115 & 1857 & 20783 & 2711 & 290.2 & 123.6 & 123.7 & -- \\
\hline

 & 29361 & 1657 & 30267 & 1705 & 214.7 & 66.75 & 67.16 & -- \\* 
 & 28493 & 1526 & 29571 & 1582 & 150.4 & 69.22 & 70.12 & -- \\*
\cline{2-9}
 & 29663 & 1605 & 31406 & 1708 & 167.7 & 68.39 & 70.07 & -- \\* 
\multirow{-4}{*}{\tt 5000\_20}  & 28825 & 1481 & 30539 & 1584 & 186.8 & 70.74 & 72.86 & -- \\
\hline

\rowcolor[gray]{.8}  & 328.9 & 51.80 & 387.7 & 61.98 & -- & 0.273 & 0.313 & $\circ$ \\* 
\rowcolor[gray]{.8}  & 289.5 & 42.00 & 343.6 & 52.08 & -- & 0.316 & 0.358 & $\circ$ \\*
\cline{2-9}
\rowcolor[gray]{.8}  & 345.5 & 50.80 & 396.4 & 59.59 & 0.986 & 0.286 & 0.320 & $\circ$ \\* 
\rowcolor[gray]{.8} \multirow{-4}{*}{\tt 1000\_30}  & 308.5 & 42.14 & 355.4 & 50.57 & 1.12 & 0.328 & 0.364 & $\circ$ \\
\hline

 & 1229261 & 142458 & 1355238 & 159621 & 9448 & 3335 & 3538 & $\circ$ \\* 
 & 1104626 & 114531 & 1262710 & 141777 & 6341 & 3409 & 3622 & $\circ$ \\*
\cline{2-9}
 & 1268906 & 140585 & 1354901 & 152294 & -- & 3479 & 3607 & -- \\* 
\multirow{-4}{*}{\tt 5000\_30}  & 1151697 & 115171 & 1258607 & 133535 & -- & 3552 & 3696 & -- \\
\hline

\rowcolor[gray]{.8}  & 3727 & 768.7 & 3699 & 776.0 & 17.54 & 3.74 & 3.78 & -- \\* 
\rowcolor[gray]{.8}  & 3110 & 494.0 & 3059 & 497.8 & 14.38 & 4.21 & 4.14 & -- \\*
\cline{2-9}
\rowcolor[gray]{.8}  & 3908 & 777.7 & 3906 & 794.6  & 12.33 & 3.91 & 3.98 & -- \\* 
\rowcolor[gray]{.8} \multirow{-4}{*}{\tt 1000\_40}  & 3304 & 520.4 & 3289 & 538.8 & 13.93 & 4.40 & 4.36 & -- \\
\hline

 & 949.9 & 238.7 & 849.5 & 211.9 & 3.92 & 0.686 & 0.644 & $\star$ \\* 
 & 738.9 & 162.3 & 653.1 & 136.2 & 3.05 & 0.838 & 0.771 & $\star$ \\*
\cline{2-9}
 & 968.8 & 232.3 & 973.1 & 236.5 & 1.75 & 0.711 & 0.731 & -- \\* 
\multirow{-4}{*}{\tt 500\_50}  & 769.7 & 162.6 & 768.0 & 163.1 & 2.05 & 0.868 & 0.874 & -- \\
\hline

\rowcolor[gray]{.8}  & 75966 & 16108 & 74069 & 15553 & 430.5 & 87.07 & 87.25 & -- \\* 
\rowcolor[gray]{.8}  & 70375 & 15366 & 68292 & 14806 & 316.1 & 95.01 & 95.10 & -- \\*
\cline{2-9}
\rowcolor[gray]{.8}  & 79196 & 16216 & 76319 & 15397 & 294.8 & 91.14 & 90.52 & -- \\* 
\rowcolor[gray]{.8} \multirow{-4}{*}{\tt 1000\_50}  & 73931 & 15515 & 70916 & 14683 & 326.3 & 99.12 & 98.36 & -- \\
\hline

 & 11912 & 3134 & 11318 & 2965 & 68.67 & 10.42 & 10.32 & -- \\* 
 & 8649 & 1973 & 8140 & 1803  & 43.60 & 12.55 & 12.06 & -- \\*
\cline{2-9}
 & 13189 & 3363 & 12640 & 3216 & 28.96 & 11.57 & 11.51 & -- \\* 
\multirow{-4}{*}{\tt 500\_60}  & 9861 & 2216 & 9378 & 2063 & 32.74 & 13.90 & 13.47 & -- \\
\hline

\rowcolor[gray]{.8}  & 1205 & 360.6 & 1104 & 325.9   & -- & 0.931 & 0.935 & -- \\* 
\rowcolor[gray]{.8}  & 836.8 & 224.5 & 757.9 & 197.8 & -- & 1.15 & 1.09 & $\star$ \\*
\cline{2-9}
\rowcolor[gray]{.8}  & 1397 & 405.3 & 1318 & 379.4   & -- & 1.07 & 1.05 & -- \\* 
\rowcolor[gray]{.8} \multirow{-4}{*}{\tt 300\_65}  & 1008 & 264.4 & 942.9 & 242.4 & -- & 1.33 & 1.29 & -- \\
\hline

 & 65250 & 17831 & 62803 & 17193 & -- & 62.30 & 62.16 & -- \\* 
 & 50299 & 14163 & 47488 & 13239 & -- & 73.04 & 71.74 & -- \\*
\cline{2-9}
 & 72288 & 19074 & 72730 & 19320 & -- & 69.25 & 71.68 & -- \\* 
\multirow{-4}{*}{\tt 500\_65}  & 57809 & 15643 & 57619 & 15642 & -- & 80.83 & 82.29 & -- \\
\hline

\rowcolor[gray]{.8}  & 3726 & 1152 & 3067 & 932.2  & 25.07 & 3.04 & 2.67 & $\star$ \\* 
\rowcolor[gray]{.8}  & 2466 & 720.3 & 1925 & 519.3 & 13.08 & 3.79 & 3.25 & $\star$ \\*
\cline{2-9}
\rowcolor[gray]{.8}  & 3881 & 1156 & 4234 & 1284   & 7.04 & 3.21 & 3.58 & $\circ$ \\* 
\rowcolor[gray]{.8} \multirow{-4}{*}{\tt 300\_70}  & 2660 & 732.4 & 2917 & 831.1 & 7.88 & 4.00 & 4.39 & $\circ$ \\
\hline

 & 492014 & 143505 & 465416 & 135411 & 3562 & 521.4 & 512.1 & -- \\* 
 & 334038 & 97140 & 312075 & 89662  & 1677 & 597.0 & 582.6 & -- \\*
\cline{2-9}
 & 574019 & 162193 & 536902 & 151032 & 1427 & 610.6 & 594.2 & -- \\* 
\multirow{-4}{*}{\tt 500\_70}  & 412843 & 117343 & 381648 & 107110 & 1417 & 699.2 & 671.1 & -- \\
\hline

\rowcolor[gray]{.8}  & 1187 & 422.3 & 1172 & 424.5   & 13.40 & 0.971 & 0.976 & -- \\* 
\rowcolor[gray]{.8}  & 619.0 & 208.6 & 602.5 & 207.3 & 4.90 & 1.23 & 1.22 & -- \\*
\cline{2-9}
\rowcolor[gray]{.8}  & 1469 & 511.8 & 1458 & 516.2   & 2.59 & 1.18 & 1.20 & -- \\* 
\rowcolor[gray]{.8} \multirow{-4}{*}{\tt 200\_80}  & 833.0 & 274.4 & 820.7 & 275.8 & 2.75 & 1.54 & 1.55 & -- \\
\hline

 & 135194 & 46638 & 107192 & 36505 & 1377 & 143.5 & 120.7 & $\star$ \\* 
 & 69045 & 23123 & 52842 & 17118  & 429.4 & 161.3 & 132.8 & $\star$ \\*
\cline{2-9}
 & 169699 & 57512 & 152912 & 51628 & 296.6 & 180.1 & 170.4 & $\star$ \\* 
\multirow{-4}{*}{\tt 300\_80}  & 95146 & 31331 & 83660 & 27077 & 290.8 & 206.7 & 192.6 & $\star$ \\
\hline

\rowcolor[gray]{.8}  & 17527 & 6956 & 13913 & 5501 & 705.2 & 21.03 & 17.59 & $\star$ \\* 
\rowcolor[gray]{.8}  & 5738 & 2407 & 4493 & 1879  & 80.66 & 21.89 & 17.88 & $\star$ \\*
\cline{2-9}
\rowcolor[gray]{.8}  & 22919 & 9091 & 20974 & 8393 & 51.89 & 27.85 & 26.33 & $\star$ \\* 
\rowcolor[gray]{.8} \multirow{-4}{*}{\tt 200\_90}  & 9620 & 3885 & 8772 & 3573 & 41.38 & 32.52 & 30.39 & $\star$ \\
\hline

 & 13572 & 6072 & 8592 & 3778 & 1386 & 23.40 & 15.69 & $\star$ \\* 
 & 2811 & 1429 & 1666 & 836.1 & 64.31 & 18.99 & 12.06 & $\star$$\star$ \\*
\cline{2-9}
 & 13901 & 6515 & 9125 & 4211  & 126.6 & 25.00 & 17.25 & $\star$ \\* 
\multirow{-4}{*}{\tt 200\_95}  & 4563 & 2249 & 3016 & 1468 & 70.74 & 27.38 & 19.14 & $\star$ \\
\hline

\rowcolor[gray]{.8}  & 207249 & 101342 & 115876 & 55630 & 308379 & 868.2 & 513.8 & $\star$$\star$ \\* 
\rowcolor[gray]{.8}  & 16007 & 9236 & 9045 & 5174  & 2859 & 358.0 & 207.6 & $\star$$\star$ \\*
\cline{2-9}
\rowcolor[gray]{.8}  & 181164 & 97100 & 72846 & 38568   & -- & 752.8 & 329.9 & $\star$$\star$$\star$ \\* 
\rowcolor[gray]{.8} \multirow{-4}{*}{\tt 300\_98}  & 36249 & 21256 & 15543 & 9107 & -- & 926.6 & 401.8 & $\star$$\star$$\star$ \\
\hline

 & 9999 & 5460 & 5475 & 2867   & -- & 95.88 & 56.35 & $\star$$\star$ \\* 
 & 140.9 & 88.76 & 53.91 & 33.06 & 42.51 & 22.51 & 7.56 & $\star$$\star$$\star$ \\*
\cline{2-9}
 & 21386 & 15619 & 1055 & 723.0 & -- & 191.0 & 10.73 & $\star$$\star$$\star$ \\* 
\multirow{-4}{*}{\tt 500\_994}  & 710.0 & 512.7 & 113.7 & 81.42 & -- & 226.7 & 32.02 & $\star$$\star$$\star$ \\
\hline

\rowcolor[gray]{.8}  & 5002 & 3186 & 2491 & 1589 & -- & 153.8 & 78.99 & $\star$$\star$ \\* 
\rowcolor[gray]{.8}  & 9.61 & 7.40 & 6.04 & 4.49 & 50.14 & 12.82 & 9.46 & $\star$ \\*
\cline{2-9}
\rowcolor[gray]{.8}  & 13315 & 10803 & 5816 & 4841  & -- & 394.5 & 177.7 & $\star$$\star$$\star$ \\* 
\rowcolor[gray]{.8} \multirow{-4}{*}{\tt 1000\_998} & 185.8 & 151.2 & 92.64 & 75.04 & -- & 803.3 & 337.2 & $\star$$\star$$\star$ \\
\hline
\hline
\end{longtable}}

\section{Concluding Remarks}
\label{sec:conc}

In this paper, we propose a new Russian Dolls Search
algorithm, improving another implementation by \"Osterg\aa rd \cite{Oste02}
in several directions like the use of approximate colorings for subproblems 
pruning, an effective use of bit-level parallelism, and the application of an
enhanced elimination rule. These improvements allow the algorithm to further
reduce the running times of the faster previously published combinatorial algorithms in several 
instances. The computational experiments aiming at checking whether the proposed
algorithm is competitive with respect to the more efficient ones in the
literature were accomplished. Results show the effectiveness of the combination
of techniques employed in \Call{RDMC}{} for hard instances (graphs with a high density). In particular, for graphs
of density beyond 0.8, our algorithm is more than twice faster in
several graphs tested. These results show that, for some combinatorial
optimization problems, the Russian Dolls method can constitute a very interesting alternative to classical
Branch and Bound approaches.

\bibliographystyle{plain}
\bibliography{./clqdolls}

\end{document}